\documentclass[11pt,oneside,letterpaper]{article}
\usepackage{eurosym}
\usepackage{graphicx,amsmath,amsfonts}
\usepackage{setspace}
\usepackage{amssymb}
\usepackage{dsfont}
\usepackage{amsmath}
\usepackage{amsfonts}
\usepackage{ifthen}
\usepackage{mathtools}
\usepackage{graphicx}
\usepackage{pgfplots}
\usepackage{enumerate}
\usepackage{color}
\usepackage{framed}
\usepackage[margin=1in]{geometry}
\usepackage{natbib}
\usepackage{hyperref}
\usepackage{fp}
\usepackage[capitalize]{cleveref}
\usepackage{amsthm}
\usepackage{caption}
\usepackage{subcaption}
\usepackage{palatino}
\usepackage{cuted,balance}
\usepackage{amssymb}
\usepackage{soul}
\usepackage{bbm}
\usepackage{graphicx,amsmath,amsfonts}

\usepackage[ruled]{algorithm2e}

\newtheorem{assumption}{Assumption}

\newtheorem{proposition}{Proposition}
\newtheorem{theorem}{Theorem}
\newtheorem{lemma}{Lemma}
\newtheorem{claim}{Claim}
\newtheorem{remark}{Remark}
\newtheorem{example}{Example}

\allowdisplaybreaks

\let\originaleqref\eqref
\renewcommand{\eqref}{\originaleqref}
\begin{document}

\title{Contract Design With Safety Inspections}
\author{ Alireza Fallah \thanks{
Simons Laufer Mathematical Sciences Institute (Mathematical Sciences Research Institute)} $^\text{\textdagger}$ \and Michael I. Jordan \thanks{University of California, Berkeley} }
\date{\today}
\maketitle

\begin{abstract}
{
We study the role of regulatory inspections in a contract design problem in which a principal interacts separately with multiple agents. Each agent's hidden action includes a dimension that determines whether they undertake an extra costly step to adhere to safety protocols. The principal's objective is to use payments combined with a limited budget for random inspections to incentivize agents towards safety-compliant actions that maximize the principal's utility.
We first focus on the single-agent setting with linear contracts and present an efficient algorithm that characterizes the optimal linear contract, which includes both payment and random inspection. We further investigate how the optimal contract changes as the inspection cost or the cost of adhering to safety protocols vary. Notably, we demonstrate that the agent's compensation increases if either of these costs escalates. However, while the probability of inspection decreases with rising inspection costs, it demonstrates nonmonotonic behavior as a function of the safety action costs.
Lastly, we explore the multi-agent setting, where the principal's challenge is to determine the best distribution of inspection budgets among all agents. We propose an efficient approach based on dynamic programming to find an approximately optimal allocation of inspection budget across contracts. We also design a random sequential scheme to determine the inspector's assignments, ensuring each agent is inspected at most once and at the desired probability. Finally, we present a case study illustrating that a mere difference in the cost of inspection across various agents can drive the principal's decision to forego inspecting a significant fraction of them, concentrating its entire budget on those that are less costly to inspect.
}
\end{abstract}

\section{Introduction}
The rapid growth of data-oriented applications has led to a surge in new products and services offered by companies that provide personalized user experiences. Alongside the benefits of personalization, however, there are growing concerns about potential unwanted side effects that are not necessarily revealed or declared by these companies. For example, many users wonder whether their data is stored securely every time they enter their sensitive information on an online platform, especially given that instances of data breaches and cyberattacks have become common.

In the realm of healthcare, the failure of companies to reveal negative and undesired side effects of drugs has had devastating consequences, as in the Vioxx case in the early 2000s and the opioid epidemic exacerbated by Purdue Pharma in the United States. These episodes demonstrate how inadequate disclosure and transparency can have severe consequences for public health.

Another concerning example pertains to the design practices of leading tech companies. There are valid concerns about whether they pay enough attention to safety, security, and reliability when creating new products and algorithms. Releasing cutting-edge technologies without adequate testing or safeguards raises questions about potential risks to users and the wider public.

Beyond simply documenting such concerns, the question arises as to how we might incentivize platforms to follow safety and security measures. In this work, we propose a mechanism design framework based on contract theory to provide such incentives, focusing on an incentive-producing role for \textit{inspections}. Building on the classical principal-agent model, where the agents (in this case, the platforms) take hidden and costly actions that result in a reward for the principal, we introduce an additional dimension to the model. Specifically, each agent must also decide whether to take a costly step to adhere to safety and security measures or disregard them. In the event that the agent neglects these measures, negative side effects may occur with a non-negligible probability, leading to adverse consequences for the principal.

The principal offers compensation to each agent based on the reward generated from the principal's action. Moreover, the principal retains the option to conduct a random and costly inspection, revealing whether the agent has complied with the safety and security measures without disclosing any information about the agent's underlying action. { In addition, although the principal maintains a separate contract with each agent, the collective design of these contracts is constrained by the principal's limited inspection budget; in particular the number of inspectors. In essence, the principal's objective is to identify the optimal set of contracts, consisting of payments and inspections, that maximizes its overall utility subject to its inspection budget.

Our work introduces a framework to model and analyze the role of regulatory and partial inspections in contract design. These inspections aim to confirm the agent's compliance with laws emphasizing societal considerations, such as safety. Our first set of results focuses on the design of linear contracts in a single-agent setting. We characterize the optimal probability of inspection as a function of the ratio of the reward paid to the agent, and show that it is a piecewise convex and decreasing function. Our detailed characterization enables us to design an algorithm that finds the optimal linear contract in quasi-linear time in the number of actions.

Furthermore, we determine that when the inspection becomes more costly for the principal, they tend to reduce the inspection probability but compensate the agent more to ensure they are motivated to adhere to safety regulations. Similarly, if the cost associated with observing safety protocols rises for the agent, the principal modifies the optimal contract by increasing the agent's compensation to guarantee adherence to safe practices. Surprisingly, in such scenarios, the optimal probability of inspection does not necessarily increase. In fact, it turns out that the increase in payment might even allow the principal to decrease the inspection level. This observation underscores that merely ramping up regulatory inspection isn't always the most effective strategy, especially in contexts where safety comes at a high cost.

Next, we turn our attention to the multi-agent setting. Here, due to their budget constraints, the principal cannot assign the optimal inspection level to each agent's contract and must thus determine the best allocation of inspections. We demonstrate that the principal's utility from each contract, when considered as a function of the maximum permissible inspection, is a piecewise concave and weakly increasing function. Further, we establish that the principal's problem in this context is closely related to the multiple-choice knapsack problem. We then introduce a dynamic-programming-based algorithm that finds an $\epsilon$-approximate solution, with a time complexity that is polynomial in terms of the number of agents, the number of actions, and $1/\epsilon$. We further use a random procedure to assign inspectors to agents, ensuring that each agent is inspected with the targeted probability of inspection. It's important to note that we cannot determine each inspector's schedule independently since each agent must be inspected no more than once. Accordingly, our procedure determines the assignment of inspectors sequentially. For every pair of consecutive inspectors, there is at most one common agent they might inspect. Nevertheless, we design each inspector's assignment distribution based on the preceding inspector to prevent overlaps and ensure that the desired inspection level is attained.

Finally, we present a case study illustrating an intriguing dynamic: even when dealing with identical agents who solely differ in terms of the cost incurred by the principal to inspect them, the principal may decide to abstain from inspecting a substantial proportion of those with higher inspection costs. This observation suggests that agents have the potential leverage to influence the principal's action by increasing the cost barriers to monitor their adherence to safety protocols. In fact, by doing so, they can dissuade the principal from monitoring them and hence increase their welfare.}
\subsection{Related Work}
{
Our model builds on the hidden-action principal-agent model \citep{grossman1992analysis}. Within this general framework, the study of costly state verification dates back to the work of \citet{townsend1979optimal, gale1985incentive} for debt contracts (see also chapter 5.3 in \cite{bolton2004contract} for discussion on costly verification or disclosure). 

Our paper aligns most closely with the literature on contract design with random monitoring \citep{jost1991monitoring, jost1996role, strausz1997delegation, barbos2022optimal}. In \cite{jost1991monitoring} the principal decides randomly to monitor the agent's action through a costly inspection, and it is argued that in the optimal contract, the principal either performs the inspection or pays the inspection cost to the agent. The work in \cite{jost1996role} considers a model where the principal has private information regarding their monitoring costs. In \cite{barbos2022optimal}, the inspection uncovers the agent's exact action with a certain probability, and in other cases, provides no information. Our work departs from the existing literature in three primary respects: 
First, we introduce a model of partial inspection in contract theory. In our framework, the inspection fully discloses adherence to safety standards but reveals no information about the other dimension of the agent' actions. Second, we establish the computational complexity of both equilibrium characterization and its comparative statics. And third, our results span both single-agent and multi-agent settings.

It is also worth mentioning that while our work focuses on one-period contract design, the optimal randomized inspection in dynamic contracts over time has been considered in the literature as well \citep{varas2020random, chen2020optimal, ball2023should, orlov2022frequent}.

In mechanism design, our paper relates to the literature on mechanism design with costly inspection or verification \citep{ben2014optimal, mylovanov2017optimal, li2020mechanism}. The focus in this line of work is on the principal's problem of allocating an object based on the reports from the agents of their private types, reports which are subject to potential inspection by the principal. These papers along with others such as \cite{erlanson2020costly, halac2020commitment} use inspection as a tool when monetary payments are not feasible; our work, on the other hand, allows for both payments and inspections. This brings our work closer to \cite{alaei2020optimal} in the mechanism design literature, where the assumption is that an auctioneer can use both payments and a (full yet deferred) inspection.

Our work is also related to the literature on partial or probabilistic verification in mechanism design \citep{green1986partially, ball2019probabilistic, caragiannis2012mechanism, ferraioli2018probabilistic}. In these models, an agent has a private type, and inspections can differentiate certain type pairs from each other but might be ineffective with others. Our safety inspection model can also be viewed in this model: pairs comprising a safe action and an unsafe action are discernible, whereas pairs consisting of two safe or two unsafe actions remain indistinct. In these works, however, inspections are presumed to be cost-free, and their primary focus is to characterize the class of social choice functions that can be implemented truthfully.

Our motivations are similar to those in the literature on regulatory inspections for incentivizing companies to adhere to standard policies \citep{harrington1988enforcement, choe1999compliance, ferraro2008asymmetric}. For instance, in \cite{harrington1988enforcement}, companies are partitioned into two groups. One group is inspected more frequently, and its members face steeper fines if found to be violating protocols. Adhering to standard protocols incurs a cost, and firms can be shifted from one group to another---either via a reward or a punishment---based on their performance.

Our work also relates to the literature on computational aspects of contract theory \citep{babaioff2006combinatorial, dutting2022combinatorial, dutting2021complexity, zhu2022sample}.  As in this line of work, we focus on the class of linear contracts, given their simplicity and interpretability, and the fact that linear contracts have been shown to be robust to the unknown actions \citep{carroll2015robustness} or unknown distributions \citep{dutting2019simple}.

}
\section{The Model}
We consider a multi-agent setting with $m$ agents and one principal. For agent $\ell \in [m]{:= \{1, \cdots, m\}}$, we use the notation $(a^\ell,s^\ell)$ to represent the agent's action, where $a^\ell \in \mathcal{A}^\ell = \{a^\ell_1, \cdots, a^\ell_n\}$ determines the effort the agent invests in providing its service,\footnote{Here, to simplify the notation, we assume all agents have $n$ actions. However, our analysis will remain the same when they have different numbers of actions.} and $s^\ell \in \{0,1\}$ indicates whether the agent considers safety measures. Accordingly, we call actions with $s^\ell=1$ and $s^\ell=0$ as \textit{safe} and \textit{unsafe} actions, respectively. 
For any $i \in [n]$, the cost of action $a^\ell_i$ is denoted by $c^\ell_i \geq 0$. Additionally, the cost of complying with safety measures (i.e., $s^\ell=1$) is denoted by $\kappa^\ell_S \geq 0$. Hence, the total cost of action $(a^\ell_i,s^\ell)$ is given by $c^\ell_i + \mathbbm{1}_{s^\ell=1} \kappa^\ell_S$.

When agent $\ell$ takes action $(a^\ell,s^\ell)$, a random reward $r^\ell \in \mathcal{R}^\ell \subseteq \mathbb{R}^{\geq 0} \cup \{-\infty\}$ accrues to the principal. Specifically, if the agent takes action $(a^\ell,1)$ which includes adherence to safety protocols, a nonnegative reward $r^\ell$ is generated with probability $f^\ell(r^\ell|a^\ell)$. Conversely, if the agent takes action $(a^\ell,0)$, neglecting safety measures, there is a probability $\alpha^\ell$ that negative side effects occur, leading to a reward of $-\infty$. With the complementary probability, these side effects do not materialize, and a nonnegative reward $r^\ell$ is generated with probability $f^\ell(r^\ell|a^\ell)$, similar to the case where the agent followed the safety measures. 

The contract between each agent and the principal comprises two elements. The principal, upon observing the reward $r^\ell$ from agent $\ell$, compensates them for their action by paying them $t^\ell(r^\ell) \geq 0$ with the condition $t^\ell(-\infty) = 0$.  In other words, the principal pays nothing if side effects occur. Additionally, the principal has the option to perform an inspection, with probability $\beta^\ell$ and at a cost of $\kappa^\ell_I$, before the reward realization, which reveals whether the agent adhered to safety measures or not, i.e., it reveals the value of $s^\ell$. What connects all of the contracts is the fact that the principal has { a limited} capacity for only { $B \in \mathbb{Z}^+$} unit of inspection, meaning that the condition $\sum_{\ell=1}^m \beta^\ell \leq { B}$ should hold {(one can think of this as having $B$ inspectors available)}. We also assume that the two events of inspection and occurrence of side effects are independent. 

We denote the expected reward to the principal and the expected payment to the $\ell$-th agent as a result of agent $\ell$ taking action $(a_i^\ell,1)$ by $R^\ell_i$ and $T^\ell_i$, respectively, i.e., 
\begin{equation}
R^\ell_i := \mathbb{E}_{r \sim f^\ell(.|a^\ell_i)}[r]
\quad \text{and} \quad 
T^\ell_i := \mathbb{E}_{r \sim f^\ell(.|a^\ell_i)}[t^\ell(r)].
\end{equation}
\subsection{Implementable Actions}
We denote agent $\ell$'s expected utility when taking action $(a^\ell_i,s^\ell)$ by $\mathcal{U}^\ell_a(a^\ell_i,s^\ell)$, which is defined as follows:
\begin{equation}
\mathcal{U}^\ell_a(a^\ell_i,s^\ell) = 
\begin{cases}
T^\ell_i - c^\ell_i - \kappa^\ell_S & s^\ell=1, \\  
(1-\beta^\ell)(1-\alpha^\ell) T^\ell_i - c^\ell_i & s^\ell=0.
\end{cases}
\end{equation}
We say an action $(a^\ell,s^\ell)$ is \textit{implementable} for agent $\ell$ if there exists a contract, consisting of $(t^\ell(\cdot), \beta^\ell)$, such that the following two conditions hold:
\begin{itemize}
\item \textit{Incentive compatibility (IC)}: the agent has no incentive to deviate and choose another action, i.e., 
$\mathcal{U}^\ell_a(a^\ell,s^\ell) \geq \mathcal{U}^\ell_a(a',s')$, for any other $a' \in \mathcal{A}^\ell$ and $s' \in \{0,1\}$.
\item \textit{Individual rationalism (IR)}: the agent is not better off by not taking the contract at all, i.e., $\mathcal{U}^\ell_a(a^\ell,s^\ell) \geq 0 $.
\end{itemize}
In other words, IC and IR together ensure that, the best response of the agent to the offered contract is to take the action $(a^\ell,s^\ell)$. We let $\mathcal{I}^\ell$ denote the set of implementable actions for agent $\ell$. 

It is worth noting that having $(a^\ell,s^\ell) \in \mathcal{I}^\ell$ for all $\ell \in [m]$ does not necessarily imply that the tuple of actions $((a^\ell,s^\ell))_{\ell=1}^m$ is implementable for all agents simultaneously, as we have not yet taken into account the constraint $\sum_{\ell=1}^m \beta^\ell \leq B$. To this end, we also say a tuple of actions $((a^\ell,s^\ell))_{\ell=1}^m$ is \textit{fully implementable} (and denote the set of such tuples by $\mathcal{I}$) if, for any $\ell$, $(a^\ell,s^\ell)$ is implementable by some contract $(t^\ell(\cdot), \beta^\ell)$ such that $\sum_{\ell=1}^m \beta^\ell \leq B$.
\subsection{The Principal's Problem}
Let us denote the principal's expected utility from agent $\ell$ when that agent takes action $(a^\ell_i,s^\ell)$ by $\mathcal{U}^\ell_p(a^\ell_i,s^\ell)$. For $s^\ell=0$, this utility is $-\infty$ as the side effects arise with a non-zero probability and lead to a reward of $-\infty$. Hence, the principal would strictly prefer safe actions over unsafe ones. For the case of $s^\ell=1$, the expected utility of principal is given by
\begin{equation}
\mathcal{U}^\ell_p(a^\ell_i,1) = 
R^\ell_i - T^\ell_i - \beta^\ell \kappa^\ell_I.
\end{equation}
In addition, we denote the total expected utility of the principal when agents take actions $((a_i^\ell,s_i^\ell))_{\ell=1}^m$ by $\mathcal{U}_p \left (((a_i^\ell,s_i^\ell))_{\ell=1}^m \right)$, defined as follows:
\begin{equation}
\mathcal{U}_p \left (((a_i^\ell,s_i^\ell))_{\ell=1}^m \right ) =
\begin{cases}
\sum_{\ell=1}^m \mathcal{U}_p^\ell(a^\ell_i,1) 
& \text{ if } s^\ell_i=1 \text{ for all } \ell, \\
-\infty & \text{ otherwise.}
\end{cases} 
\end{equation}

Now, the principal's problem can be seen as designing a contract that incentivizes agents to play the tuple of actions that maximizes her expected utility among all implementable tuples of actions. In other words, the principal's problem can be cast as finding the contract that implements the solution to the following maximization problem:
\begin{align}
\max_{((a_i^\ell,s_i^\ell))_{\ell=1}^m} & \mathcal{U}_p \left (((a_i^\ell,s_i^\ell))_{\ell=1}^m \right)  \\
\text{s.t.} & \quad ((a_i^\ell,s_i^\ell))_{\ell=1}^m \in \mathcal{I}. \nonumber 
\end{align}

We make the following assumptions throughout the paper:
\begin{assumption} \label{assumption:no-dominant}
For every agent, we assume different actions have different costs, and moreover, an action with a higher cost also has a higher expected reward. Also, we assume that no two actions of any given agent have the same cost.      
\end{assumption}
Assumption \ref{assumption:no-dominant} ensures that no action dominates another one, meaning that it leads to a higher expected reward at a lower cost. Under this assumption, and without loss of generality, we assume $c^\ell_1 < c^\ell_2 < \cdots < c^\ell_n$ and $R^\ell_1 < R^\ell_2 < \cdots < R^\ell_n$ for every $\ell \in [m]$.  
\begin{assumption} \label{assumption:safety_feasible}
For every agent $\ell$, we have $\max_{i} (R^\ell_i - c^\ell_i) > \kappa^\ell_S$.    
\end{assumption}
Note that if this assumption does not hold for an agent, it implies that, even if that agent receives the entire reward as payment, their utility would remain nonpositive for any safe action. In simpler terms, this assumption guarantees there is a way to make at least one safe action implementable for every agent.

Our focus is on linear contracts, which are contracts for which the agent's payment is directly proportional to the reward received by the principal. General contracts can often be complex and challenging to understand or implement, while linear contracts offer a simpler and more practical alternative. Moreover, under reasonable assumptions, linear contracts are known to be worst-case optimal \cite{dutting2019simple}.  In our setting, focusing on linear contracts allow us to better isolate the role of partial inspection in contract design over a useful and intuitive class of contracts. 

More formally, we consider payment function $t^\ell(r) = \gamma^\ell r$, for some $\gamma^\ell \in [0,1]$ chosen by the principal. Hence, the principal has two sets of parameters to choose: $\gamma^\ell$, the ratio of the reward paid to agent $\ell$, and $\beta^\ell$, the probability of performing inspection regarding the safety of agent $\ell$'s action. Note also that in this case, $T^\ell_i = \gamma^\ell R^\ell_i$.

\section{The Single-Agent Setting}
To gain a better understanding of the nature of the problem, we begin with the single-agent setting, i.e., $m=1$. {Without loss of generality, we assume $B=1$ in this case.} To simplify the notation, we drop the superscripts throughout this section. 

Let us first consider what happens if the principal is not allowed to do the inspection, i.e., $\beta$ is set to $0$. In this case, the principal should intuitively offer a higher payment ratio to persuade the agent to adhere to safety protocols. However, as the following result shows, this may not be enough.
\begin{lemma} \label{lemma:inspection_needed}
If $\alpha < \frac{\kappa_S}{R_n}$, then there is no safe action that is implementable by a linear contract without inspection.     
\end{lemma}
{The proof (along with other omitted proofs) can be found in the appendix.} This result shows that, when the negative side effects are rare enough, the principal would need to perform the inspection to keep the agent committed to the safety measures. Otherwise, when the occurrence probability of side effect is small enough, the agent takes a chance in not abiding the safety protocols. 
Next, we continue by characterizing the properties of the linear contract. By IC constraint, if a linear contract $(\gamma, \beta)$ implements the safe action $(a_i,1)$, then $\mathcal{U}_a(a_i,1) \geq \mathcal{U}_a(a_j,1)$ for any $j \neq i$. This simplifies to $\gamma R_i - c_i \geq \gamma R_j -c_j$ for any $j \neq i$. 

Now, inspired by \cite{dutting2019simple}, we develop the following geometric characterization: for any $i \in [n]$, let us define the linear function $h_i:[0,1] \to \mathbb{R}$ as $h_i(\gamma) = \gamma R_i -c_i$. As we stated above, if the action $(a_i,1)$ is implementable by $(\gamma, \beta)$, then $h_i(\gamma) \geq h_j(\gamma)$ for any $j \neq i$. As a result, to find the set of implementable actions, we need to characterize the upper envelope of the set of functions $\{h_j(\cdot)\}_{j=1}^n$, denoted by $u_h(\cdot)$, which is a piecewise linear and increasing function (see Figure \ref{fig:upperenvelope} for an example).
\begin{figure}
\centering
\begin{minipage}{.5\textwidth}
  \centering
  \begin{tikzpicture}
    \begin{axis}[
        xlabel=\(\gamma\),
        domain=0:1,
        legend pos=north west,
        ytick=\empty,     
        xtick=\empty,     
    ]
        \addplot[red, domain=0:0.1, thick] {x*1 - 0.1};
        \addplot[orange, domain=0.1:0.4, thick] {x*3 - 0.3};
        \addplot[green, domain=0.4:0.7, thick] {x*4 - 0.7};
        \addplot[blue, domain=0.7:1, thick] {x*6 - 2.1};
        \addlegendentry{\( h_1(\gamma) := R_1\gamma - c_1 \)}
        \addlegendentry{\( h_2(\gamma) := R_2\gamma - c_2 \)}
        \addlegendentry{\( h_3(\gamma) := R_3\gamma - c_3 \)}
        \addlegendentry{\( h_4(\gamma) := R_4\gamma - c_4 \)}
        \addplot[red, domain=0.1:1, dashed] {x*1 - 0.1};
        \addplot[orange, domain=0:0.05, dashed] {x*3 - 0.3};
        \addplot[orange, domain=0.4:1, dashed] {x*3 - 0.3};
        \addplot[green, domain=0:0.4, dashed] {x*4 - 0.7};
        \addplot[green, domain=0.7:1, dashed] {x*4 - 0.7};
        \addplot[blue, domain=0.0:0.7, dashed] {x*6 - 2.1};
        
    \end{axis}
\end{tikzpicture}
  \captionof{figure}{Illustration of $u_h(\cdot)$}
  \label{fig:upperenvelope}
\end{minipage}%
\begin{minipage}{.5\textwidth}
  \centering
  \begin{tikzpicture}
    \begin{axis}[
        xlabel=\(\gamma\),
        domain=0:1,
        ytick=\empty,     
        xtick=\empty,     
    ]
        \addplot[red, domain=0:0.1, thick] {x*1 - 0.1};
        \addplot[orange, domain=0.1:0.4, thick] {x*3 - 0.3};
        \addplot[green, domain=0.4:0.7, thick] {x*4 - 0.7};
        \addplot[blue, domain=0.7:1, thick] {x*6 - 2.1};
        
        \node[label={left:{\tiny $\gamma_1$}},circle,fill,inner sep=1pt] at (axis cs:0,-0.1) {};
        \node[label={right:{\tiny $\gamma_2$}},circle,fill,inner sep=1pt] at (axis cs:0.1,0) {};
        \node[label={right:{\tiny $\gamma_3$}},circle,fill,inner sep=1pt] at (axis cs:0.4,0.9) {};
        \node[label={left:{\tiny $\gamma_4$}},circle,fill,inner sep=1pt] at (axis cs:0.7,2.1) {};
        \node[label={right:{\tiny $\gamma_5$}},circle,fill,inner sep=1pt] at (axis cs:1,3.9) {};
        
        \node[label={right:{\small $(\gamma, u_h(\gamma))$}},circle,fill,inner sep=1pt] at (axis cs:0.75,2.4) {};

        \draw[->](axis cs:0.75,2.4) |- (axis cs:0.525,1.4);
        \node[label={right:{\small $\kappa_S$}}] at (axis cs:0.73,1.8) {};
        \node[label={left:{\small $(1-\beta)(1-\alpha)\gamma$}}] at (axis cs:0.525,1.4) {};

    \end{axis}
\end{tikzpicture}
  \captionof{figure}{How to characterize $\beta(\gamma)$}
  \label{fig:findbeta}
\end{minipage}
\end{figure}
Each segment of $u_h(\cdot)$ corresponds to a specific $h_i(\cdot)$, representing the dominant function within that segment. This implies that, in that segment, $(a_i,1)$ stands as the sole implementable safe action (for the $\gamma$ values pertaining to that segment), with no incentive to divert to an alternative \textit{secure} action. Notice that we have not yet considered the deviation to unsafe actions or the IR constraint.   

Next, using this derivation, we establish the following result on implementable safe actions and their corresponding linear contracts:
\begin{proposition} \label{prop:upper_envelope}
Suppose Assumptions \ref{assumption:no-dominant} and \ref{assumption:safety_feasible} hold. There exist $0=\gamma_0 < \gamma_1 < \gamma_2 < \cdots < \gamma_k < \gamma_{k+1} = 1$ such that the following holds:
\begin{enumerate}
\item No safe action is implementable for $\gamma < \gamma_1$. 
\item There exists a set of actions $i_1 < \cdots i_k$ such that, for any $j \in [k]$, the action $(a_{i_j},1)$ is implementable with $\gamma \in [\gamma_j, \gamma_{j+1}]$. Moreover, $(a_{i_j},1)$ is the only implementable safe action for $\gamma \in (\gamma_j, \gamma_{j+1})$.
\item For any $\gamma \geq \gamma_1$, there exists $\beta(\gamma)$ such that $(\gamma, \beta)$ implements a safe action if and only if $\beta \geq \beta(\gamma)$.
\end{enumerate}
\end{proposition}
\begin{proof}
First, note that, in order for a safe action $(a,1)$ to be implementable, IR should hold as well. Hence, if $u_h(\gamma) \leq \kappa_S$, no safe action would be implementable by $\gamma$. As a result, no safe action is implementable for $\gamma$ below 
\begin{equation}
\gamma_1 := (u_h)^{-1}(\kappa_S).   
\end{equation}
Also, note that $u_h(1) = \max_j (R_j - c_j)$, which by assumption is greater than $\kappa_S$. Therefore, $\gamma_1 < 1$.

Now, let us focus on $\gamma \geq \gamma_1$. As we stated earlier, $u_h(\cdot)$ is a piecewise linear and increasing function. Hence, there exist $i_1 < \cdots i_k$ and $\gamma_1 < \gamma_2 < \cdots < \gamma_k < \gamma_{k+1} = 1$ such that, for any $j \in [k]$, we have
\begin{equation}
u_h(\gamma) = h_{i_j}(\gamma) \text{ for any }
\gamma \in [\gamma_j, \gamma_{j+1}].
\end{equation}
In other words, on segment $[\gamma_j, \gamma_{j+1}]$, $u_h(\cdot)$ is equal to $h_{i_j}(\cdot)$. Hence, for $\gamma \in [\gamma_j, \gamma_{j+1}]$ the only safe action that can potentially be implemented is $(a_{i_j},1)$. What remains is to rule out deviation to unsafe actions. Recall that the agent's utility from an unsafe action $(a_v, 0)$ is given by 
\begin{equation}
(1-\beta)(1-\alpha) \gamma R_v - c_v   
\end{equation}
which is, in fact, $h_v \big ((1-\beta)(1-\alpha) \gamma \big)$. As a consequence, the maximum utility that an agent can obtain from an unsafe action is given by $\max_{v} h_v \big ((1-\beta)(1-\alpha) \gamma \big)$ which is $u_h\big ((1-\beta)(1-\alpha) \gamma \big)$. Now, the IC constraint would require us to have
\begin{equation} \label{eqn:IC_unsafe}
u_h(\gamma) - \kappa_S  \geq   u_h\big ((1-\beta)(1-\alpha) \gamma \big).
\end{equation}
Notice that the left-hand side is nonnegative, since $\gamma \geq \gamma_1$. Also, $u_h$ is an increasing function, which starts from a negative value, i.e., $u_h(0)= \max_j (-c_j)$. Hence, we can choose $\beta$ large enough such that \eqref{eqn:IC_unsafe} holds (see Figure \ref{fig:findbeta} for an illustration). 
\end{proof}
Note that, for any $j \in [k]$, any $\gamma \in [\gamma_j, \gamma_{j+1}]$ with $\beta \geq \beta(\gamma)$ makes action $(a_{i_j},1)$ implementable for the agent. Since increasing $\beta$ would only decrease the principal's utility, without loss of generality, we could assume the principal picks $\beta = \beta(\gamma)$. Next, we characterize how this probability of inspection $\beta(\gamma)$ changes as a function of $\gamma$. 
\begin{lemma}\label{lemma:gamma_increasing}
Suppose Assumptions \ref{assumption:no-dominant} and \ref{assumption:safety_feasible} hold, and recall the definition of $\beta(\gamma)$ from Proposition \ref{prop:upper_envelope}. Then, $\beta(\gamma)$ is a decreasing function of $\gamma$. Moreover, it is strictly decreasing when $\beta(\gamma) > 0$.     
\end{lemma}
This result formalizes an intuitive observation: if the principal aims to reduce agent's payment, they should increase the inspection probability; conversely, if the principal wishes to avoid expensive inspections, they should offer higher compensation to the agent, encouraging adherence to safety protocols.

A natural question arises at this point: how should the principal determine the optimal trade-off between the agent's payment and the cost of inspection? Let $\gamma$ fall within the interval $[\gamma_j, \gamma_{j+1}]$ for some $j$. Recall that, in this case, the principal's utility is given by $(1-\gamma)R_{i_j} - \beta(\gamma) \kappa_I$. Consequently, the marginal cost associated with increasing the agent's payment is $R_{i_j}$, while the marginal cost of inspection is $\kappa_I$. The next result helps us to find the appropriate $\gamma$ that balances this trade-off.
\begin{lemma} \label{lemma:gamma_convex}
Under the premise of Proposition \ref{prop:upper_envelope}, and for any $j \in [k]$, $\beta(\gamma)$ is a convex function over the interval $[\gamma_j, \gamma_{j+1}]$. Moreover, it is strictly convex when $\beta(\gamma) > 0$.   
\end{lemma}
\noindent \textbf{\textit{Proof sketch:}}
Recall that, as illustrated in Figure \ref{fig:findbeta}, $\beta(\gamma)$ is given by 
\begin{equation}
\beta(\gamma) = \max \left  \{1- \frac{u_h^{-1}(u_h(\gamma) - \kappa_S)}{\gamma(1-\alpha)}, 0 \right \}.    
\end{equation}
Now, let us determine where $\beta(\cdot)$ could be nondifferentiable.
As $\gamma$ sweeps over the interval $[\gamma_j, \gamma_{j+1}]$, the corresponding $\tilde{\gamma} := u_h^{-1}(u_h(\gamma)-\kappa_S)$ may fall in this segment $[\gamma_j, \gamma_{j+1}]$ or one of the previous segments $[\gamma_i, \gamma_{i+1}]$ for some $i <j$. For instance, in the example illustrated in Figure \ref{fig:findbeta}, for $\gamma$ marked on the plot, $\tilde{\gamma}$ falls within the previous segment. Hence, there exists a sequence $\gamma_j = \gamma_{j,0} < \gamma_{j,1} < \cdots < \gamma_{j,v_j} < \gamma_{j,v_j+1}= \gamma_{j+1}$ such that $\tilde{\gamma}$ falls within the same segment for any $\gamma \in (\gamma_{j,q}, \gamma_{j,q+1})$. See Figure \ref{fig:beta_partition} for an illustration on how the $\{\gamma_{j,q}\}$ are defined.

{ Now, for any $q$, the $\beta(\cdot)$ function is differentiable over the interval $(\gamma_{j,q}, \gamma_{j,q+1})$. To prove Lemma \ref{lemma:gamma_convex}, we first establish that $\beta(\cdot)$ is indeed convex over each interval $(\gamma_{j,q}, \gamma_{j,q+1})$ by showing that its derivative is increasing there. Finally, we present an argument detailing how the convexity of $\beta(\cdot)$ over the entire interval $[\gamma_j, \gamma_{j+1}]$ can be inferred from its derivative across these subintervals. $\square$}
\begin{figure}
\centering
\begin{minipage}{.5\textwidth}
  \centering
  \begin{tikzpicture}
    \begin{axis}[
        domain=-0.2:1,
        ytick=\empty,     
        xtick=\empty,     
        height = 190,    
        width = 230,     
    ]
        \addplot[red, domain=-0.2:0.1, thick] {x*1 - 0.1};
        \addplot[orange, domain=0.1:0.4, thick] {x*3 - 0.3};
        \addplot[green, domain=0.4:0.7, thick] {x*4 - 0.7};
        \addplot[blue, domain=0.7:1, thick] {x*6 - 2.1};


        \node[label={left:{$\gamma$}}] at (axis cs:0.5,-0.5) {};
        
        \node[label={left:{\tiny $\gamma_1$}},circle,fill,inner sep=1pt] at (axis cs:-0.2,-0.3) {};
        \node[label={right:{\tiny $\gamma_2$}},circle,fill,inner sep=1pt] at (axis cs:0.1,0) {};
        \node[label={right:{\tiny $\gamma_3$}},circle,fill,inner sep=1pt] at (axis cs:0.4,0.9) {};
        \node[label={right:{\tiny $\gamma_4$}},circle,fill,inner sep=1pt] at (axis cs:0.7,2.1) {};
        \node[label={right:{\tiny $\gamma_5$}},circle,fill,inner sep=1pt] at (axis cs:1,3.9) {};

        \draw[->](axis cs:-0.2,-0.3) |- (axis cs:0.33,0.7);
        \node[label={right:{\tiny $\gamma_{2,1}$}},circle,fill,inner sep=1pt] at (axis cs:0.33,0.7) {};
        \draw[->](axis cs:0.1,0) |- (axis cs:0.425,1);
        \node[label={right:{\tiny $\gamma_{3,1}$}},circle,fill,inner sep=1pt] at (axis cs:0.425,1) {};
        \draw[->](axis cs:0.4,0.9) |- (axis cs:0.65,1.9);
        \node[label={right:{\tiny $\gamma_{3,2}$}},circle,fill,inner sep=1pt] at (axis cs:0.65,1.9) {};
        \draw[->](axis cs:0.7,2.1) |- (axis cs:0.866,3.1);
        \node[label={right:{\tiny $\gamma_{4,1}$}},circle,fill,inner sep=1pt] at (axis cs:0.866,3.1) {};

        \node[label={left:{\tiny $\kappa_S$}}] at (axis cs:-0.15,0.3) {};

    \end{axis}
\end{tikzpicture}
  \vspace{-4mm}
  \caption{How $\{\{\gamma_{j,q}\}_{q=1}^{v_j}\}_{j=1}^k$'s are defined \label{fig:beta_partition}}
\end{minipage}%
\begin{minipage}{.5\textwidth}
  \centering
  \includegraphics[width=\linewidth]{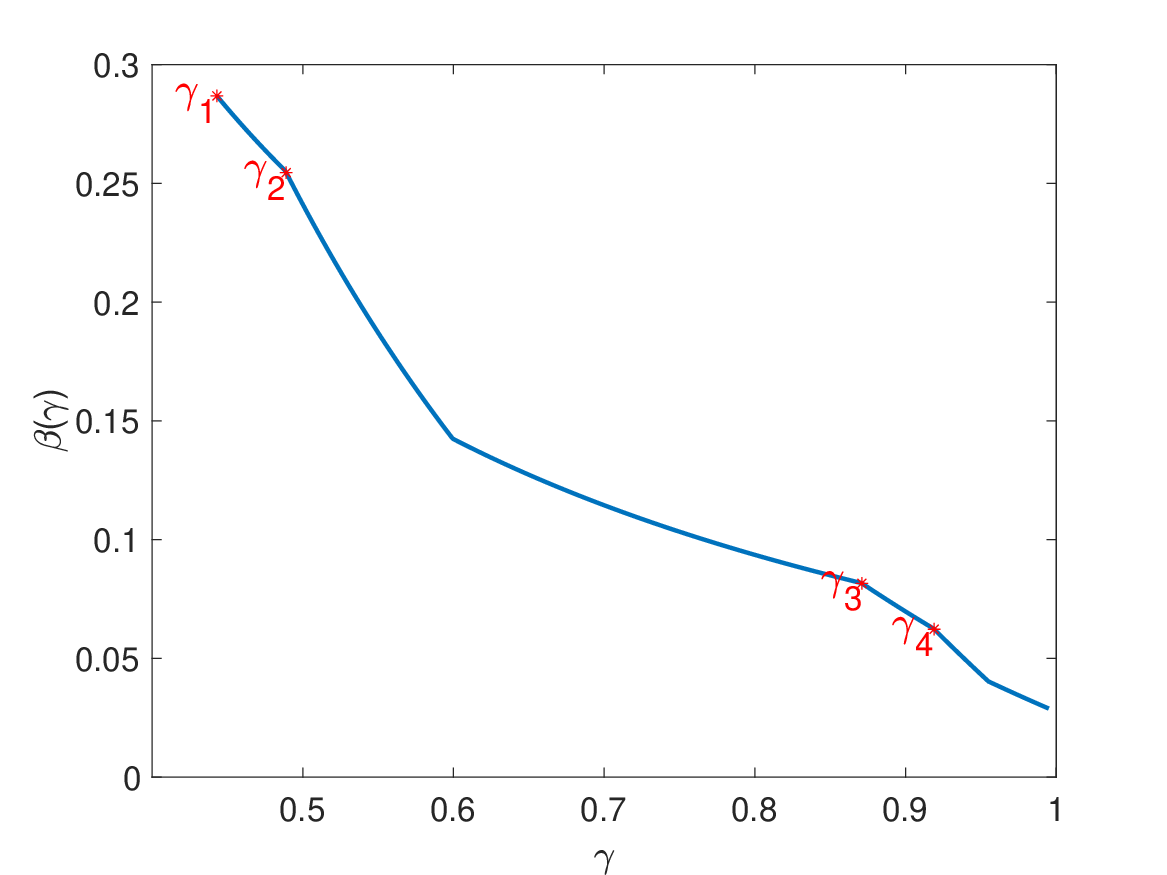}
  \caption{$\beta(\gamma)$ as a function of $\gamma$ (see Remark \ref{example:beta_gamma} for the details)}
  \label{fig:beta_gamma}
\end{minipage}
\end{figure}
\begin{remark} \label{example:beta_gamma}
It is worth noting that while $\beta(\cdot)$ is a convex function over each interval $[\gamma_j, \gamma_{j+1}]$, it is not necessarily convex over the whole interval $[\gamma_1,1]$. See Figure \ref{fig:beta_gamma} for an example with $n=6$ actions with the following parameters:
\begin{equation}
[R_i]_{i=1}^6 = [2,3,7,9,11,13], \,
[c_i]_{i=1}^6 = [1, 1.2, 2.1, 3.1, 4.8, 6.6], \, \text{and }
\kappa_I = \kappa_S = 1.
\end{equation}
\end{remark}
Now, having the characterization of the set of implementable actions, we focus on finding the optimal contract which maximizes the principal's utility. We start by showing that it can be found efficiently.  
\begin{theorem} \label{theorem:computaion_singlaAgent}
Suppose Assumptions \ref{assumption:no-dominant} and \ref{assumption:safety_feasible} hold. Then, the optimal linear contract $(\gamma^*, \beta^*)$ can be characterized in time $\mathcal{O}(n \log n)$.   
\end{theorem}
\begin{proof}
The first step is to characterize $u_h$. Notice that, using a duality argument, we could transfer the problem of finding the upper envelope function to the problem of finding the convex hull of the set of points $(R_i,c_i)_{i=1}^n$, which can be done in $\mathcal{O}(n \log n)$ using the Graham Scan Algorithm \cite{graham1972efficient}. This allows us to find the set $\{i_1, \cdots, i_k\}$ as defined in Proposition \ref{prop:upper_envelope}, and hence, the points $\gamma_1, \cdots, \gamma_k$ in time $\mathcal{O}(k \log k)$ which is bounded by $\mathcal{O}(n \log n)$. In addition, note that we could find the points $\{\{\gamma_{j,q}\}_{q=1}^{v_j}\}_{j=1}^k$ in $\mathcal{O}(n \log n)$ defined in the proof of Lemma \ref{lemma:gamma_convex}, by computing $\{u_h^{-1}(u_h(\gamma_j) + \kappa_S)\}$ (as illustrated in Figure \ref{fig:beta_partition}). Furthermore, the total number of such points, i.e., $\sum_{j=1}^k v_j$, is also $k$, so this part can also be completed in time $\mathcal{O}(n \log n)$.

Next, for any $j \in [k]$, we first find the optimal contract ($\gamma, \beta(\gamma)$), and condition on $\gamma \in [\gamma_j, \gamma_{j+1}]$. After doing so, the principal, among all such contracts, can pick the one that leads to the highest expected utility for her. Therefore, it suffices to focus on the case $\gamma \in [\gamma_j, \gamma_{j+1}]$. Recall that, in this case, the principal's utility is given by 
\begin{equation}
\mathcal{U}_p = (1-\gamma)R_{i_j} - \beta(\gamma) \kappa_I.    
\end{equation}
Note that, by Lemma \ref{lemma:gamma_convex}, $\mathcal{U}_p$ is a concave function over $[\gamma_j, \gamma_{j+1}]$. Hence, to find its maximum, we just need to find the point $\frac{d}{d \gamma} \mathcal{U}_p =0$, where $\beta(\cdot)$ is differentiable, and also check the endpoints and points of nondifferentiable points. Thus, we would need to check for the solutions of
\begin{equation} \label{eqn:maximum_zero_derivative}
\beta'(\gamma) = -\frac{R_{i_j}}{\kappa_I}.    
\end{equation}
Using the notation in the proof of Lemma \ref{lemma:gamma_convex}, we have an explicit characterization of $\beta'(\gamma)$ over all intervals $(\gamma_{j,q}, \gamma_{j,q+1})$, and so we can find the potential solution to \eqref{eqn:maximum_zero_derivative} in time $\mathcal{O}(1)$. As a result, the total computation time that we need to check all the intervals $(\gamma_{j,q}, \gamma_{j,q+1})$, for all $j$ and $0 \leq q \leq v_j$, and their endpoints, is $\mathcal{O}(n)$. This completes the proof. 
\end{proof}
A summary of the above steps is provided in Algorithm \ref{alg:optimal_singleAgent}. It is worth noting that the optimal contract is not always unique.
\begin{algorithm*}[t]
\KwIn{The reward and cost of different actions $(R_i, c_i)_{i=1}^n$}
Find the upper envelope function $u_h(\cdot)$ by finding the convex hall of $(R_i, c_i)_{i=1}^n$ using the Graham Scan Algorithm \citep{graham1972efficient};\\
Denote the convex hull by $(R_{i_j},c_{i_j})_{j=1}^k$ such that $i_1 < \cdots < i_k$; \\
Find $\gamma_1 < \cdots < \gamma_k$ as defined in Proposition \ref{prop:upper_envelope} and illustrated in Figure \ref{fig:findbeta}. \\
Find all the nondifferentiable points $\{\{\gamma_{j,q}\}_{q=0}^{v_j+1}\}_{j=1}^k$ as illustrated in Figure \ref{fig:beta_partition}; \\
Let $\mathcal{S}$ denote the set of contracts corresponding to $\{\{(\gamma_{j,q}\}\}_{q=0}^{v_j+1}\}_{j=1}^k$ and the associated utility of the principal; \\
\For{$j=1$ to $k$}{
\For{$q=1$ to $v_j$}{
If the following equation has a solution over $[\gamma_{j,q}, \gamma_{j,q+1}]$, then add it to $\mathcal{S}$:
\begin{equation*}
\frac{c_{i_j} - c_{i_{j_q}} + \kappa_S}{R_{i_{j_q}} \gamma^2(1-\alpha)} = \frac{R_{i_j}}{\kappa_I};   
\end{equation*}
}
}
\textbf{Output}: Pick the contract(s) in $\mathcal{S}$ that maximizes the principal's utility
\caption{Computing the optimal contract in the single-agent setting}
	\label{alg:optimal_singleAgent}
\end{algorithm*}

With the computational guarantees for determining the optimal contracts in hand, we next investigate the comparative statics of the optimal contract, showing how the agent's payment and the probability of inspection vary as the parameters of the model corresponding to the inspection costs change. Note that these results do not require any assumption on the uniqueness of the optimal contract. 
\begin{theorem} \label{theorem:singleAgent_comparativestatics}
Suppose Assumptions \ref{assumption:no-dominant} and \ref{assumption:safety_feasible} hold. Let $(\gamma^*, \beta^*)$ be an optimal contract. 
\begin{enumerate}
\item Suppose the principal's cost of inspection $\kappa_I$ increases, and let $(\gamma'^*, \beta'^*)$ denote an optimal contract under this new setting. Then, we have $\gamma'^* \geq \gamma^*$ and $\beta'^* \leq \beta^*$.
\item Suppose the agent's cost of complying with safety measure $\kappa_S$ increases, and let $(\gamma''^*, \beta''^*)$ denote an optimal contract under this new setting. Then, we have $\gamma''^* \geq \gamma^*$.
\end{enumerate}
\end{theorem}
In essence, the first part of Theorem \ref{theorem:singleAgent_comparativestatics} shows that an increase in the inspection pricing $\kappa_I$ prompts the principal to reduce the inspection probability. Simultaneously, the principal increases the agent's payment to ensure the agent remains incentivized to observe safety protocols. The second part of the theorem shows that when the agent has to pay higher costs for adhering to safety protocols, the principal increases the agent's payment, thereby motivating continued adherence to safety protocols. In this case, and at first glance, one might intuitively presume that the principal would also increase the inspection probability, based on the same rationale. However, as highlighted in the next example, this is not necessarily the case.
\begin{example}
Consider a setting where agent has $n=6$ actions with rewards and costs given by
\begin{equation}
[R_i]_{i=1}^6 = [1.5, 3,4,6,7,9], \quad 
[c_i]_{i=1}^6 = [1, 1.3, 1.5, 2.5, 3.4, 5.2].
\end{equation}
Figure \ref{fig:optimalcontract_pi_S} illustrates the optimal share of payment to the agent $\gamma^*$ and the probability of inspection $\beta^*$ as we increase $\kappa_S$. In particular, as Figure \ref{fig:gamma_pi_S} shows $\gamma^*$ is a (weakly) increasing function of $\kappa_S$ which is aligned with our result in Theorem \ref{theorem:singleAgent_comparativestatics}. On the other hand, as we discussed above and Figure \ref{fig:beta_pi_S} shows, $\beta^*$ is not necessarily a monotone function of the cost $\kappa_S$. { The underlying cause of this phenomenon is the interplay of two opposing forces. On one hand, by increasing the cost of safety for the agent, i.e., $\kappa_S$, the beta function $\beta(\cdot)$ increases pointwise. Put another way, had the payment ratio $ \gamma$ remained unchanged, the probability of inspection would have gone up. However, as we established earlier, the optimal $\gamma^*$ at equilibrium may also increase. In this context, an increased payment implies that we reduce the probability of inspection. The combined impact of these two forces dictates whether the optimal probability of inspection $\beta^*$ ascends or descends as a function of $\kappa_S$.}
\begin{figure}
\centering
\begin{subfigure}{.5\textwidth}
  \centering
  \includegraphics[width=\linewidth]{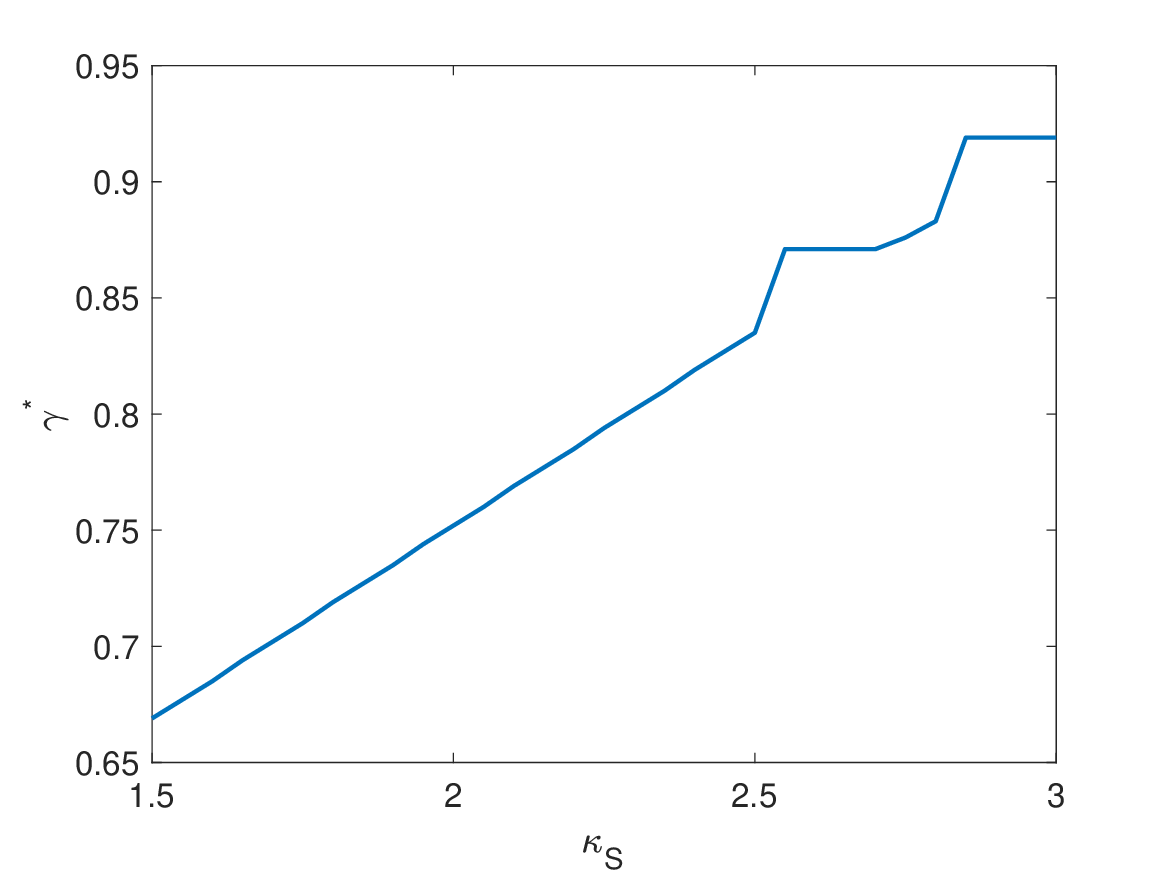}
  \caption{$\gamma^*$ as a function of $\kappa_S$}
  \label{fig:gamma_pi_S}
\end{subfigure}%
\begin{subfigure}{.5\textwidth}
  \centering
  \includegraphics[width=\linewidth]{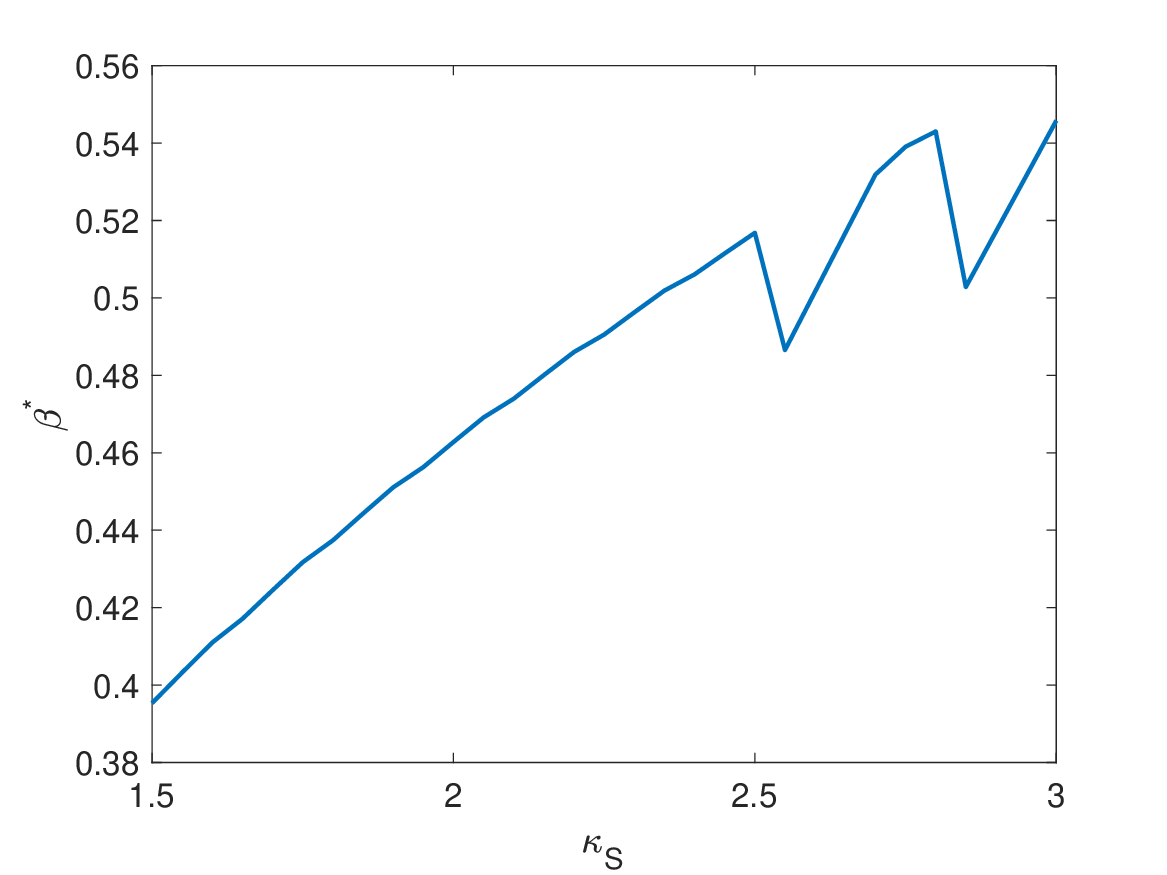}
  \caption{$\beta^*$ as a function of $\kappa_S$}
  \label{fig:beta_pi_S}
\end{subfigure}
\caption{An example depicting how the optimal contract changes as $\kappa_S$ increases.}
\label{fig:optimalcontract_pi_S}
\end{figure}
\end{example}
\section{The Multi-Agent Setting}
We turn to the multi-agent setting. Although we understand how to determine the optimal contract for each agent, we must ensure that the cumulative probability of inspection does not exceed {the inspection budget}. For every agent \( \ell \in [m] \), let \( \beta_{\min}^\ell \) represent the lowest possible probability of inspection across all contracts that implement one of the actions for agent \( j \). As established by Lemma \ref{lemma:gamma_increasing}, the probability of inspection diminishes as the agent's payment increases. Thus, it reaches its minimum when the agent receives the maximum payment, which is the total reward. Consequently, we deduce that \( \beta_{\min}^\ell = \beta^\ell(1) \), where \( \beta^\ell(\cdot) \) is the \( \beta(\cdot) \) function introduced in Proposition \ref{prop:upper_envelope} in relation to agent \( \ell \).\footnote{We reuse the major notation from the single-agent setting by using superscripts to distinguish among the different agents. In particular, recalling Proposition \ref{prop:upper_envelope}, and for any agent $\ell \in [m]$, we have $0 < \gamma_1^\ell < \gamma_2^\ell < \cdots< \gamma_{k^\ell}^\ell < \gamma_{k^\ell+1}^\ell =1$, where for any $j \in [k^\ell]$, the action $(a_{i_j^\ell}^\ell,1)$ is implementable for agent $\ell$ with $\gamma \in [\gamma_j^\ell, \gamma_j^{\ell+1}]$. } The subsequent assumption ensures the presence of at least one feasible set of contracts.
\begin{assumption}\label{assumption:sum_minimum_beta}
We have $\sum_{\ell=1}^m \beta_{\min}^\ell \leq { B}$.    
\end{assumption}
Using Theorem \ref{theorem:computaion_singlaAgent}, we can find optimal contracts for different agents separately and do so in time \(\mathcal{O}(mn \log n)\). However, the total probability of inspection could potentially exceed {$B$}. In such cases, compromises would be necessary, meaning we would have to consider suboptimal contracts that can be implemented with a lower probability of inspection than the optimal ones. Consequently, the following natural question arises:
\textit{For any agent $\ell \in [m]$, and given some \(\bar{\beta}^\ell\), which contract maximizes the principal's utility $\mathcal{U}_p^\ell$ among all those whose probability of inspection is at most \(\bar{\beta}^\ell\)?}

Notice that, with $\gamma \in [\gamma_j^\ell, \gamma_{j+1}^\ell]$, the utility of principal from agent $\ell$'s action is given by
\begin{equation} \label{eqn:utility_function_gamma}
\mathcal{U}_p^\ell(\gamma, \beta^\ell(\gamma)) = (1-\gamma)R_{i_j^\ell}^\ell - \beta^\ell(\gamma) \kappa^\ell_I.     
\end{equation}
For the sake of our analysis, we find it more convenient to interpret the principal's utility as a function of $\beta$. Let us denote $\beta(\gamma_j^\ell)$ by $\beta_j^\ell$. Notice that, since $\beta^\ell(\cdot)$ is a decreasing function, we have
\begin{equation}
\beta_1^\ell > \cdots > \beta_{k^\ell}^\ell > \beta_{k^\ell+1}^\ell = \beta_{\min}^\ell.    
\end{equation}
We also denote the inverse of $\beta^\ell(\cdot)$ by $\gamma^\ell(\cdot)$. Using Lemma \ref{lemma:gamma_increasing} and \ref{lemma:gamma_convex}, it is straightforward to verify that $\gamma^\ell(\cdot)$ is also a decreasing function and it is convex over each interval $[\beta_{j+1}^\ell, \beta_j^\ell]$.

Next, we can rewrite the principal's utility from agent $\ell$'s action given in \eqref{eqn:utility_function_gamma} as a function of $\beta$. In particular, for any $\beta \in [\beta_{j+1}^\ell, \beta_j^\ell]$, we have
\begin{equation} \label{eqn:U_p_ell_beta}
\mathcal{U}_p^\ell ( \gamma^\ell(\beta), \beta) = (1-\gamma^\ell(\beta) )R_{i_j^\ell}^\ell - \beta \kappa^\ell_I.     
\end{equation}
This is a concave function over each interval $[\beta_{j+1}^\ell, \beta_j^\ell]$. Figure \ref{fig:principal_utility_beta_gamma} depicts this function for the example provided in Remark \ref{example:beta_gamma}. The dashed lines highlight the intervals $[\beta_{j+1}^\ell, \beta_j^\ell]$'s.

Now, going back to our question above, with a slight abuse of notation, we denote the maximum utility that the principal can obtain from agent $\ell$ given the condition $\beta \leq \bar{\beta}$ by $\mathcal{U}_p^\ell(\bar{\beta})$ which is given by
\begin{equation}\label{eqn:principal_utility_bar_beta}
\mathcal{U}_p^\ell(\bar{\beta}) = \max_{\beta \leq \bar{\beta}} \mathcal{U}_p^\ell ( \gamma^\ell(\beta), \beta).    
\end{equation}
It is straightforward to see that this function is (weakly) increasing. Moreover, since $\mathcal{U}_p^\ell( \gamma^\ell(\beta), \beta)$ is concave over each interval $[\beta_{j+1}^\ell, \beta_j^\ell]$, the function $\mathcal{U}_p^\ell(\bar{\beta})$ consists of segments that are either constant or both concave and increasing.
Figure \ref{fig:principal_utility_ironed} illustrates this function $\mathcal{U}_p^\ell(\bar{\beta})$ for the same example of Figure \ref{fig:principal_utility_beta_gamma}. 
\begin{figure}
\centering
\begin{subfigure}{.5\textwidth}
  \centering
  \includegraphics[width=\linewidth]{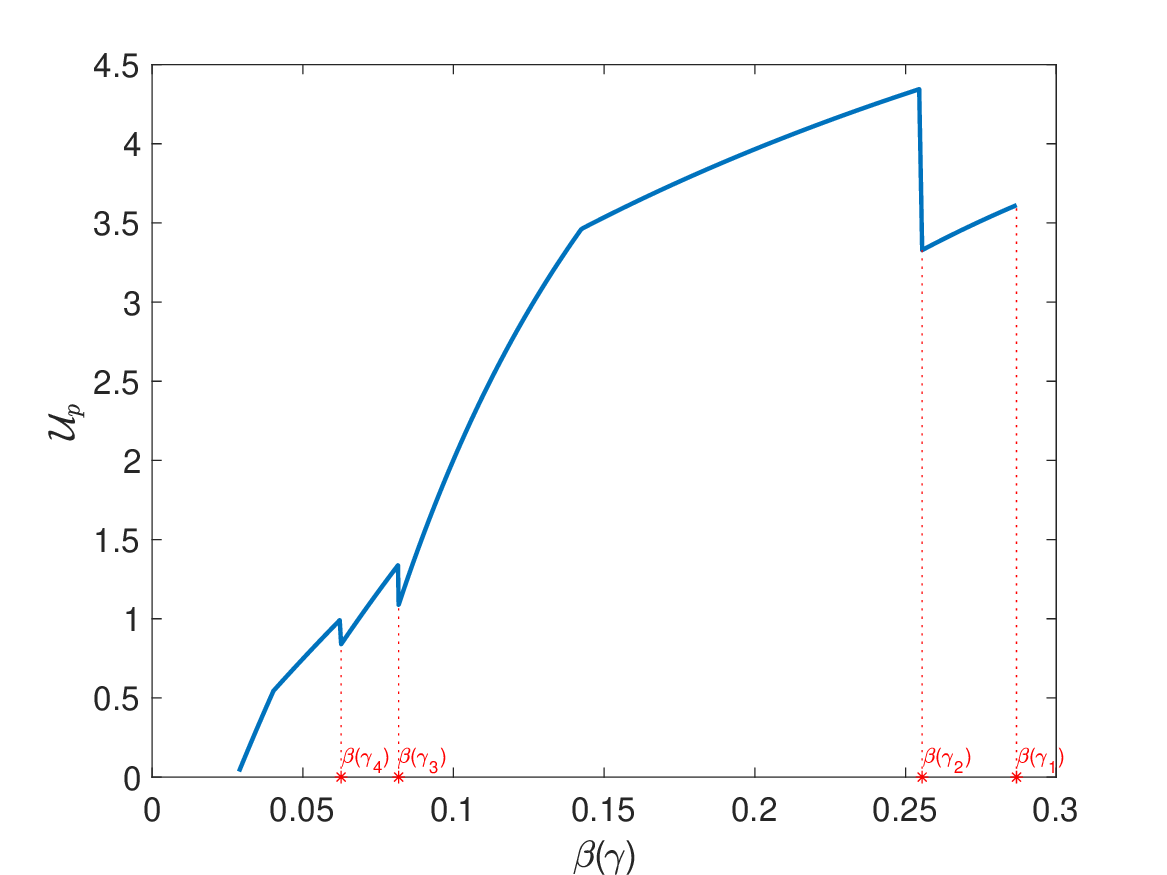}
  \caption{$\mathcal{U}_p^\ell( \gamma^\ell(\beta), \beta)$}
  \label{fig:principal_utility_beta_gamma}
\end{subfigure}%
\begin{subfigure}{.5\textwidth}
  \centering
  \includegraphics[width=\linewidth]{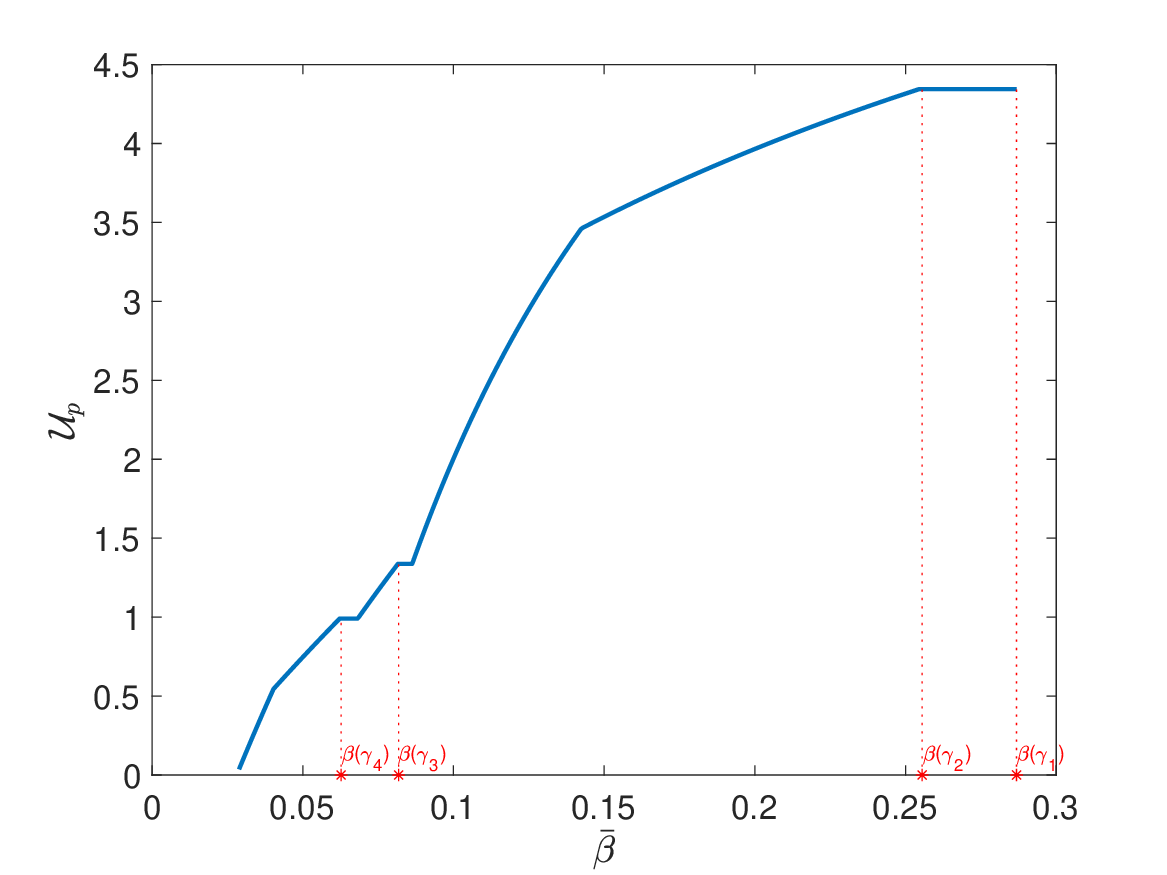}
  \caption{$\mathcal{U}_p^\ell(\bar{\beta})$}
  \label{fig:principal_utility_ironed}
\end{subfigure}
\caption{Principal's utility for agent's $\ell$ action}
\label{fig:optimalcontract_pi_S}
\end{figure}

Now, the principal's problem can be formulated as
\begin{equation}\label{eqn:principal_problem_multiAgent}   
\max_{(\bar{\beta}^\ell)_{\ell=1}^m} \sum_{\ell=1}^m \mathcal{U}_p^\ell(\bar{\beta}^\ell) \,
\text{ such that } \sum_{\ell=1}^m \bar{\beta}^\ell \leq B.
\end{equation}
{ The optimization problem \eqref{eqn:principal_problem_multiAgent} bears a resemblance to the multiple-choice knapsack problem (MCKP), a variant of the classic knapsack problem. In the MCKP, items are categorized into classes, and the goal is to maximize the cumulative value of the chosen items without surpassing the knapsack's capacity and under the constraint of selecting at most one item from each class. To make the connection to our setting clearer, consider discretizing each function $\mathcal{U}_p^\ell(\cdot)$ and grouping all the samples into a single class. For a given class $ \ell \in [m]$, each item takes the form $(\beta, \mathcal{U}_p^\ell(\beta))$, where the probability of inspection $\beta$ is seen as the weight and  $\mathcal{U}_p^\ell(\beta)$ is interpreted as this item’s value. The knapsack's capacity in our scenario is set to $B$, representing the total inspection budget. The constraint of selecting one item from each class translates to the constraint that we can inspect each agent at most once.

A comprehensive survey of various methods to address the MCKP can be found in \cite{kellerer2004multiple}. We choose to use a dynamic programming approach which is inspired by the algorithm presented in \cite{dudzinski1987exact} for the MCKP. This yields the following complexity result.
}
\begin{theorem} \label{theorem:multi-agent-DP}
Suppose Assumptions \ref{assumption:no-dominant}-\ref{assumption:sum_minimum_beta} hold. Then, for any $\epsilon > 0$, an $\epsilon$-approximate solution to the principal's problem \eqref{eqn:principal_problem_multiAgent} can be found in time $\mathcal{O}\left ( (mn + \frac{m^3 B}{\epsilon^2})\log n \right )$.    
\end{theorem}
\begin{proof}
{ 
First recall that, similar to Algorithm \ref{alg:optimal_singleAgent}, we can characterize the function $\mathcal{U}_p^\ell(\cdot)$ in time $\mathcal{O}(n \log n)$ for any $\ell$, and hence, in time $\mathcal{O}(m n \log n)$ for all $\ell \in [m]$. Consequently, we can compute $\mathcal{U}_p^\ell(\beta)$ at any given $\beta$ in time $\mathcal{O}(\log n)$.

To ensure each agent receives the minimum inspection, we rewrite the optimization problem \eqref{eqn:principal_problem_multiAgent} as
\begin{equation}\label{eqn:principal_problem_multiAgent_1}   
\max_{(x^\ell)_{\ell=1}^m} \sum_{\ell=1}^m \mathcal{U}_p^\ell(\beta_{\min}^\ell + x^\ell) \,
\text{ such that } 
\sum_{\ell=1}^m x^\ell \leq \tilde{B}:= B - \sum_{\ell=1}^m \beta_{\min}^\ell.
\end{equation}
We next discretize the inspection levels with stepsize $\delta >0$. Let us denote the grid corresponding to agent $\ell$ by $\mathcal{G}^\ell$. 
For any $l \in [m]$ and nonegative $j \leq 
\frac{\tilde{B}}{\delta}$, let $V(l, j)$ be the solution to the following maximization problem:
\begin{equation}\label{eqn:principal_problem_multiAgent_2}   
V(l, j):= \max_{(x^\ell \in \mathcal{G}^\ell)_{\ell=1}^l} \sum_{\ell=1}^l \mathcal{U}_p^\ell(\beta_{\min}^\ell + x^\ell) +
\sum_{\ell=l+1}^m \mathcal{U}_p^\ell(\beta_{\min}^\ell)
\,
\text{ such that } 
\sum_{\ell=1}^l x^\ell \leq j \delta.
\end{equation}
In other words, $V(l, j)$ represents the highest utility the principal can achieve when searching over the grid, provided they only consider the first $l$ agents for any inspections beyond the minimum and allocate only $j\delta$ from their additional inspection budget.

Note that $V(l,j)$ admits the following recursive characterization:
\begin{equation} \label{eqn:recursive}
V(l,j) = \max_{\eta =0, \cdots, \min\{j, 1/\delta\}} 
\left ( V(l-1, j-\eta) + \mathcal{U}_p^l(\beta_{\min}^l + \eta \delta) - 
\mathcal{U}_p^l(\beta_{\min}^l)
\right ).
\end{equation}
Using \eqref{eqn:recursive}, $V(l,j)$ can be computed in time $\mathcal{O}(\log n/\delta)$. As a result, we can compute $V(m, \lfloor \tilde{B}/\delta \rfloor)$ in time $\mathcal{O}(m \tilde{B}/\delta^2 \log n)$ (and by going back recursively, we find the corresponding optimal probability of inspection for each agent). Finally, we bound the error of such a discretization.
\begin{lemma} \label{claim:error_discretization}
Let $\text{OPT}$ denote the solution to \eqref{eqn:principal_problem_multiAgent_1}. Then, the solution found over the grid $\prod_{\ell=1}^m \mathcal{G}^\ell$ with stepsize $\delta$ is lower bounded by
\begin{equation}
\text{OPT} - \delta \left ( \sum_{\ell=1}^m \left [ \frac{(R_n^\ell)^2}{\kappa^\ell_S}-\kappa_I^\ell \right ] \right ).    
\end{equation}
\end{lemma}
We defer the proof of the lemma to the appendix. Given this result, by setting 
\begin{equation}
\delta = \epsilon ~\mathcal{U}_p^1 \bigl(B - \sum_{\ell=2}^m \beta_{\min}^\ell \bigr) \left ( m \max_\ell \frac{(R_n^\ell)^2}{\kappa^\ell_S} \right )^{-1},    
\end{equation}
in which $\mathcal{U}_p^1(B - \sum_{\ell=2}^m \beta_{\min}^\ell)$ serves as a lower bound for the $\text{OPT}$, we obtain the desired approximation.
}
\end{proof}
{
Having obtained the approximately optimal solution, a natural question arises regarding its implementation: Given a specific allocation of the inspection budgets $(\bar{\beta}^\ell)_{\ell=1}^m$, how should the $B$ inspectors be (randomly) allocated among the $m$ agents to ensure that agent $\ell$ is inspected with probability $\bar{\beta}^\ell$? When $B=1$, the solution is straightforward: one can generate a uniform random variable $U$ over $[0,1]$ and inspect agent $\ell$ if $U$ falls within the interval $[\sum_{i=1}^{\ell-1} ~\bar{\beta}^i, \sum_{i=1}^{\ell} ~\bar{\beta}^i]$. However, when $B$ is greater than one, the situation becomes more complex because we must ensure that each agent is inspected by at most one inspector. We can formulate this problem within the framework of the well-known Birkhoff-von-Neumann algorithm \citep{birkhoff1946tres} by constructing a matrix where each entry represents the probability that a specific agent is inspected by a particular inspector. This algorithm presents a method to decompose an \(m \times m\) bistochastic matrix into a convex combination of permutation matrices, with a time complexity of \(\mathcal{O}(m^2)\). However, in our setting, there are no constraints on the joint distribution of inspectors and agents, apart from the provided marginals. This allows us to derive an intuitive and simpler algorithm that runs in \(\mathcal{O}(m)\) time complexity. 
\begin{lemma} \label{lemma:scheduling}
For any given vector of $(\bar{\beta}^\ell)_{\ell=1}^m$, there exists a random algorithm to assign inspectors to agents in time $\mathcal{O}(m)$ such that each agent $\ell \in [m]$ is inspected with probability $\bar{\beta}^\ell$.  
\end{lemma}
\noindent \textbf{\textit{Proof sketch:}}
Consider the first inspector who wants to allocate their one unit of inspection across agents. They start with agent one, dedicating a \(\bar{\beta}^1\) fraction of their inspection budget, and then proceed to agent two, continuing in this manner until their budget is over. This procedure, however, implies that the final agent they inspected (let us call it agent \(\ell\)) might be inspected at a probability lower than the target, \(\bar{\beta}^\ell\), due to the depletion of their inspection unit. Consequently, the second inspector should start from this agent, taking care of the residual inspection probability, and then advance to subsequent agents. 

It is critical to ensure that agent \(\ell\) undergoes inspection by no more than a single inspector.  We design a joint inspection plan for the initial two inspectors ensuring that, in each instance, at most one of them inspects agent \(\ell\). In essence, the inspection schedule of the second inspector is contingent on the actions of the first. The second inspector is allowed to inspect agent \(\ell\) only when the first inspector had inspected one of the initial \(\ell-1\) agents. The scheduling for subsequent inspectors is designed similarly, ensuring that in cases where two consecutive inspectors are in charge of inspecting one agent, they do not conduct the inspection simultaneously. The details of this method are provided in the appendix for the sake of completeness. $\square$
\subsection{An Illustrative Example}
In general, the structure of the optimal allocation of the inspection budget across agents could be complex or even counter-intuitive. For example, even when we have a set of homogeneous agents, the optimal inspection allocation does not necessarily involve inspecting all agents equally. To illustrate, consider Figure \ref{fig:principal_utility_ironed}. Imagine we have two identical agents whose corresponding principal utilities are depicted in this figure. Specifically, the utility of the principal from one agent when inspected with probability $\beta_{\min}+0.1$ is above 3, which is more than twice the utility from one agent when inspected with probability $\beta_{\min} + 0.05$, a value below 1.3. This suggests that, given two homogeneous agents with these characteristics, inspecting both equally could be suboptimal compared to inspecting one at the minimum level and allocating all the extra inspection budget to the other.

We next provide a case study that illuminates the phenomenon where the principal might opt to forgo inspecting a subset of agents, especially when there is a disparity in inspection costs. Imagine a scenario in which agents are identical in all respects except for their inspection costs. Specifically, assume that the inspection costs for agents take on one of two distinct values: high and low. Half of the agents have high inspection costs, while the other half have low costs. In this case, we show the following result:
\begin{proposition} \label{proposition:illustrative}
Given the scenario described above and assuming $\beta_{\min}^\ell = 0$ for all $\ell \in [m]$, there exists a threshold $M$ such that for every $m \geq M$, the principal does not inspect half of the agents associated with a higher inspection cost. Specifically, in any optimal solution to the principal's optimization problem \eqref{eqn:principal_problem_multiAgent}, these agents' allocated level of inspection is zero.
\end{proposition}
The proof can be found in the appendix. This outcome emphasizes that as the number of agents increases, the principal may choose to not monitor a substantial subset of agents that entail high inspection costs, directing its entire inspection budget towards those that are less costly to inspect. Conversely, this suggests that agents who manage to elevate the inspection costs for the principal might escape the inspection and, as a result, secure a higher payment.
}
\section{Conclusion}
We study the role of safety and regulatory inspections in contract design problems. In particular, we consider a principal who can use random and costly inspections, in addition to payments, to incentivize agents to take safe actions. For the single-agent setting, we provide an efficient algorithm to find the optimal linear contract, and also establish how the payment fraction and inspection probability vary as the costs of inspection or adherence to safety protocols increase. We extend our results to the multi-agent setting, where we draw connections with the Knapsack problem to find an approximately optimal set of contracts in that case. Our case study on the structure of the solution illustrates how agents who are most costly to inspect may escape monitoring.

We believe our framework can be used and extended to study further problems surrounding regulatory actions in contract design. In particular, one interesting future direction would be to study the dynamic setting where the principal interacts with the agents across multiple rounds and must decide how often to inspect different agents.
\section{Acknowledgment}
The authors thank Alex Teytelboym, Ali Makhdoumi, and Azarakhsh Malekian for insightful discussions and comments. Alireza Fallah acknowledges support by the National Science Foundation under grant number DMS-1928930 and by the Alfred P.\ Sloan Foundation under grant G-2021-16778, during his residence at the Simons Laufer Mathematical Sciences Institute in Berkeley, California, during the Fall 2023 semester. Michael Jordan acknowledges support from the Mathematical Data Science program of the Office of Naval Research under grant number N00014-21-1-2840 and the European Research Council Synergy Program.
\bibliography{References}

\begin{thebibliography}{37}
\providecommand{\natexlab}[1]{#1}
\providecommand{\url}[1]{\texttt{#1}}
\expandafter\ifx\csname urlstyle\endcsname\relax
  \providecommand{\doi}[1]{doi: #1}\else
  \providecommand{\doi}{doi: \begingroup \urlstyle{rm}\Url}\fi

\bibitem[Alaei et~al.(2020)Alaei, Belloni, Makhdoumi, and
  Malekian]{alaei2020optimal}
S.~Alaei, A.~Belloni, A.~Makhdoumi, and A.~Malekian.
\newblock Optimal auction design with deferred inspection and reward.
\newblock \emph{Available at SSRN 3700525}, 2020.

\bibitem[Babaioff et~al.(2006)Babaioff, Feldman, and
  Nisan]{babaioff2006combinatorial}
M.~Babaioff, M.~Feldman, and N.~Nisan.
\newblock Combinatorial agency.
\newblock In \emph{Proceedings of the 7th ACM Conference on Electronic
  Commerce}, pages 18--28, 2006.

\bibitem[Ball and Kattwinkel(2019)]{ball2019probabilistic}
I.~Ball and D.~Kattwinkel.
\newblock Probabilistic verification in mechanism design.
\newblock In \emph{Proceedings of the 2019 ACM Conference on Economics and
  Computation}, pages 389--390, 2019.

\bibitem[Ball and Knoepfle(2023)]{ball2023should}
I.~Ball and J.~Knoepfle.
\newblock Should the timing of inspections be predictable?
\newblock \emph{arXiv preprint arXiv:2304.01385}, 2023.

\bibitem[Barbos(2022)]{barbos2022optimal}
A.~Barbos.
\newblock Optimal contracts with random monitoring.
\newblock \emph{International Journal of Game Theory}, 51\penalty0
  (1):\penalty0 119--154, 2022.

\bibitem[Ben-Porath et~al.(2014)Ben-Porath, Dekel, and Lipman]{ben2014optimal}
E.~Ben-Porath, E.~Dekel, and B.~L. Lipman.
\newblock Optimal allocation with costly verification.
\newblock \emph{American Economic Review}, 104\penalty0 (12):\penalty0
  3779--3813, 2014.

\bibitem[Bertsekas(1997)]{bertsekas1997nonlinear}
D.~P. Bertsekas.
\newblock Nonlinear programming.
\newblock \emph{Journal of the Operational Research Society}, 48\penalty0
  (3):\penalty0 334--334, 1997.

\bibitem[Birkhoff(1946)]{birkhoff1946tres}
G.~Birkhoff.
\newblock Tres observaciones sobre el algebra lineal.
\newblock \emph{Univ. Nac. Tucuman, Ser. A}, 5:\penalty0 147--154, 1946.

\bibitem[Bolton and Dewatripont(2004)]{bolton2004contract}
P.~Bolton and M.~Dewatripont.
\newblock \emph{Contract Theory}.
\newblock MIT press, 2004.

\bibitem[Caragiannis et~al.(2012)Caragiannis, Elkind, Szegedy, and
  Yu]{caragiannis2012mechanism}
I.~Caragiannis, E.~Elkind, M.~Szegedy, and L.~Yu.
\newblock Mechanism design: from partial to probabilistic verification.
\newblock In \emph{Proceedings of the 13th acm conference on electronic
  commerce}, pages 266--283, 2012.

\bibitem[Carroll(2015)]{carroll2015robustness}
G.~Carroll.
\newblock Robustness and linear contracts.
\newblock \emph{American Economic Review}, 105\penalty0 (2):\penalty0 536--563,
  2015.

\bibitem[Chen et~al.(2020)Chen, Sun, and Xiao]{chen2020optimal}
M.~Chen, P.~Sun, and Y.~Xiao.
\newblock Optimal monitoring schedule in dynamic contracts.
\newblock \emph{Operations Research}, 68\penalty0 (5):\penalty0 1285--1314,
  2020.

\bibitem[Choe and Fraser(1999)]{choe1999compliance}
C.~Choe and I.~Fraser.
\newblock Compliance monitoring and agri-environmental policy.
\newblock \emph{Journal of agricultural economics}, 50\penalty0 (3):\penalty0
  468--487, 1999.

\bibitem[Dudzi{\'n}ski and Walukiewicz(1987)]{dudzinski1987exact}
K.~Dudzi{\'n}ski and S.~Walukiewicz.
\newblock Exact methods for the knapsack problem and its generalizations.
\newblock \emph{European Journal of Operational Research}, 28\penalty0
  (1):\penalty0 3--21, 1987.

\bibitem[D{\"u}tting et~al.(2019)D{\"u}tting, Roughgarden, and
  Talgam-Cohen]{dutting2019simple}
P.~D{\"u}tting, T.~Roughgarden, and I.~Talgam-Cohen.
\newblock Simple versus optimal contracts.
\newblock In \emph{Proceedings of the 2019 ACM Conference on Economics and
  Computation}, pages 369--387, 2019.

\bibitem[Dutting et~al.(2021)Dutting, Roughgarden, and
  Talgam-Cohen]{dutting2021complexity}
P.~Dutting, T.~Roughgarden, and I.~Talgam-Cohen.
\newblock The complexity of contracts.
\newblock \emph{SIAM Journal on Computing}, 50\penalty0 (1):\penalty0 211--254,
  2021.

\bibitem[D{\"u}tting et~al.(2022)D{\"u}tting, Ezra, Feldman, and
  Kesselheim]{dutting2022combinatorial}
P.~D{\"u}tting, T.~Ezra, M.~Feldman, and T.~Kesselheim.
\newblock Combinatorial contracts.
\newblock In \emph{2021 IEEE 62nd Annual Symposium on Foundations of Computer
  Science (FOCS)}, pages 815--826. IEEE, 2022.

\bibitem[Erlanson and Kleiner(2020)]{erlanson2020costly}
A.~Erlanson and A.~Kleiner.
\newblock Costly verification in collective decisions.
\newblock \emph{Theoretical Economics}, 15\penalty0 (3):\penalty0 923--954,
  2020.

\bibitem[Ferraioli and Ventre(2018)]{ferraioli2018probabilistic}
D.~Ferraioli and C.~Ventre.
\newblock Probabilistic verification for obviously strategyproof mechanisms.
\newblock \emph{arXiv preprint arXiv:1804.10512}, 2018.

\bibitem[Ferraro(2008)]{ferraro2008asymmetric}
P.~J. Ferraro.
\newblock Asymmetric information and contract design for payments for
  environmental services.
\newblock \emph{Ecological Economics}, 65\penalty0 (4):\penalty0 810--821,
  2008.

\bibitem[Gale and Hellwig(1985)]{gale1985incentive}
D.~Gale and M.~Hellwig.
\newblock Incentive-compatible debt contracts: The one-period problem.
\newblock \emph{The Review of Economic Studies}, 52\penalty0 (4):\penalty0
  647--663, 1985.

\bibitem[Graham(1972)]{graham1972efficient}
R.~L. Graham.
\newblock An efficient algorithm for determining the convex hull of a finite
  planar set.
\newblock \emph{Info. Proc. Lett.}, 1:\penalty0 132--133, 1972.

\bibitem[Green and Laffont(1986)]{green1986partially}
J.~R. Green and J.-J. Laffont.
\newblock Partially verifiable information and mechanism design.
\newblock \emph{The Review of Economic Studies}, 53\penalty0 (3):\penalty0
  447--456, 1986.

\bibitem[Grossman and Hart(1992)]{grossman1992analysis}
S.~J. Grossman and O.~D. Hart.
\newblock An analysis of the principal-agent problem.
\newblock In \emph{Foundations of Insurance Economics: Readings in Economics
  and Finance}, pages 302--340. Springer, 1992.

\bibitem[Halac and Yared(2020)]{halac2020commitment}
M.~Halac and P.~Yared.
\newblock Commitment versus flexibility with costly verification.
\newblock \emph{Journal of Political Economy}, 128\penalty0 (12):\penalty0
  4523--4573, 2020.

\bibitem[Harrington(1988)]{harrington1988enforcement}
W.~Harrington.
\newblock Enforcement leverage when penalties are restricted.
\newblock \emph{Journal of Public Economics}, 37\penalty0 (1):\penalty0 29--53,
  1988.

\bibitem[Jost(1991)]{jost1991monitoring}
P.-J. Jost.
\newblock Monitoring in principal-agent relationships.
\newblock \emph{Journal of Institutional and Theoretical Economics
  (JITE)/Zeitschrift f{\"u}r die gesamte Staatswissenschaft}, pages 517--538,
  1991.

\bibitem[Jost(1996)]{jost1996role}
P.-J. Jost.
\newblock On the role of commitment in a principal--agent relationship with an
  informed principal.
\newblock \emph{Journal of Economic Theory}, 68\penalty0 (2):\penalty0
  510--530, 1996.

\bibitem[Kellerer et~al.(2004)Kellerer, Pferschy, Pisinger, Kellerer, Pferschy,
  and Pisinger]{kellerer2004multiple}
H.~Kellerer, U.~Pferschy, D.~Pisinger, H.~Kellerer, U.~Pferschy, and
  D.~Pisinger.
\newblock The multiple-choice knapsack problem.
\newblock \emph{Knapsack Problems}, pages 317--347, 2004.

\bibitem[Li(2020)]{li2020mechanism}
Y.~Li.
\newblock Mechanism design with costly verification and limited punishments.
\newblock \emph{Journal of Economic Theory}, 186:\penalty0 105000, 2020.

\bibitem[Mylovanov and Zapechelnyuk(2017)]{mylovanov2017optimal}
T.~Mylovanov and A.~Zapechelnyuk.
\newblock Optimal allocation with ex post verification and limited penalties.
\newblock \emph{American Economic Review}, 107\penalty0 (9):\penalty0
  2666--2694, 2017.

\bibitem[Orlov(2022)]{orlov2022frequent}
D.~Orlov.
\newblock Frequent monitoring in dynamic contracts.
\newblock \emph{Journal of Economic Theory}, 206:\penalty0 105550, 2022.

\bibitem[Rudin(1986)]{rudin1987real}
W.~Rudin.
\newblock \emph{Real and Complex Analysis}.
\newblock McGraw-Hill, New York, 3rd edition, 1986.
\newblock ISBN 978-0070542341.

\bibitem[Strausz(1997)]{strausz1997delegation}
R.~Strausz.
\newblock Delegation of monitoring in a principal-agent relationship.
\newblock \emph{The Review of Economic Studies}, 64\penalty0 (3):\penalty0
  337--357, 1997.

\bibitem[Townsend(1979)]{townsend1979optimal}
R.~M. Townsend.
\newblock Optimal contracts and competitive markets with costly state
  verification.
\newblock \emph{Journal of Economic Theory}, 21\penalty0 (2):\penalty0
  265--293, 1979.

\bibitem[Varas et~al.(2020)Varas, Marinovic, and Skrzypacz]{varas2020random}
F.~Varas, I.~Marinovic, and A.~Skrzypacz.
\newblock Random inspections and periodic reviews: Optimal dynamic monitoring.
\newblock \emph{The Review of Economic Studies}, 87\penalty0 (6):\penalty0
  2893--2937, 2020.

\bibitem[Zhu et~al.(2023)Zhu, Bates, Yang, Wang, Jiao, and
  Jordan]{zhu2022sample}
B.~Zhu, S.~Bates, Z.~Yang, Y.~Wang, J.~Jiao, and M.~I. Jordan.
\newblock The sample complexity of online contract design.
\newblock In \emph{Proceedings of the 24th ACM Conference on Economics and
  Computation}, EC '23, page 1188, New York, NY, USA, 2023. Association for
  Computing Machinery.

\end{thebibliography}
\newpage
\renewcommand{\theequation}{\mbox{A\arabic{equation}}}
\renewcommand{\thesection}{\mbox{A\arabic{section}}}
\setcounter{equation}{0}
\setcounter{section}{0}
\renewcommand{\thelemma}{\mbox{A\arabic{lemma}}}
\setcounter{lemma}{0}
\setcounter{claim}{0}
\renewcommand{\theclaim}{\mbox{A\arabic{claim}}}
\appendix
\section{Deferred Proofs}
\subsection{Proof of Lemma \ref{lemma:inspection_needed}}
Suppose some safe action $(a_i,1)$ is implementable. Then, by IC, we should have
\begin{equation}
\mathcal{U}_a(a_i,1) \geq \mathcal{U}_a(a_i,0),    
\end{equation}
which implies
\begin{equation}
\gamma R_i - c_i - \kappa_S \geq (1-\alpha) \gamma R_i - c_i.    
\end{equation}
As a consequence, we should have
\begin{equation}
\alpha \gamma R_i \geq \kappa_S.    
\end{equation}
Using $1 \geq \gamma $ and $R_n \geq R_i$, implies that $\alpha R_n \geq \kappa_S$ which contradicts the assumption given in the lemma's statement.
\subsection{Proof of Lemma \ref{lemma:gamma_increasing}}
As depicted in Figure \ref{fig:findbeta}, we can cast $\beta(\gamma)$ as:
\begin{equation}
\beta(\gamma) = \max \left  \{1- \frac{u_h^{-1}(u_h(\gamma) - \kappa_S)}{\gamma(1-\alpha)}, 0 \right \}.    
\end{equation}
Hence, it suffices to show that
\begin{equation} \label{eqn:def_f}
f(\gamma):= \frac{u_h^{-1}(u_h(\gamma) - \kappa_S)}{\gamma(1-\alpha)}    
\end{equation}
is strictly increasing in $\gamma$. Note that $u_h(\cdot)$ is a piecewise linear function, and therefore, it is differentiable on all but finitely many points. As a result, and due to the strict monotonicity of $u_h(\cdot)$, the function $f(\cdot)$ is likewise differentiable, except at a finite number of points. Consequently, there exists a sequence $0 = \rho_1 < \rho_2 < \cdots < \rho_v = 1$ such that $f$ is differentiable within each interval $(\rho_j, \rho_{j+1})$. Furthermore, since $u_h(\cdot)$ and its inverse are both continuous, $f(\cdot)$ is also continuous. We claim that it suffices to show that the derivative of $f(\cdot)$ is positive wherever it is differentiable. By proving this, we establish that $f$ is increasing on each interval $(\rho_j, \rho_{j+1})$, which, combined with the continuity of $f$, concludes our proof.    

To show the aforementioned claim holds, note that the derivative of $f$ is given by:
\begin{align} \label{eqn:derivative_f}
f'(\gamma) = \frac{\frac{u'_h(\gamma)}{u'_h(u_h^{-1}(u_h(\gamma) - \kappa_S))} \gamma - u_h^{-1}(u_h(\gamma)-\kappa_S)}{\gamma^2(1-\alpha)}.    
\end{align}
Let us denote $u_h^{-1}(u_h(\gamma)-\kappa_S)$ by $\tilde{\gamma}$. To show the derivative \eqref{eqn:derivative_f} is positive, we need to show that $\gamma u'_h(\gamma) > \tilde{\gamma} u'_h(\tilde{\gamma})$. First, notice that, since $u_h(\cdot)$ is increasing, $\tilde{\gamma} < \gamma$. Second, notice that $u_h(\cdot)$ is convex as it is the maximum of a collection of linear (and hence convex) functions. Consequently, $u'_h(\cdot)$ is increasing, and thus, $u'_h(\gamma) > u'_h(\tilde{\gamma})$ as well, which completes the proof. 
\subsection{Proof of Lemma \ref{lemma:gamma_convex}}
Recall the definition of function $f$ from the proof of Lemma \ref{lemma:gamma_increasing}. It suffices to show $f(\cdot)$ is strictly concave over the interval $[\gamma_j, \gamma_{j+1}]$. To do so, we first show $f(\cdot)$ is strictly concave over the interval $(\gamma_{j,q}, \gamma_{j,q+1})$ for any $q$.

Let $\gamma \in (\gamma_{j,q}, \gamma_{j,q+1})$ and suppose  $\tilde{\gamma}$ falls within the interval $[\gamma_{j_q}, \gamma_{j_{q+1}}]$ where $j_q \leq j$. Recall the derivative of $f$ in \eqref{eqn:derivative_f} is given by 
\begin{equation}
f'(\gamma) = \frac{\frac{u'_h(\gamma)}{u'_h(\tilde{\gamma})} \gamma - \tilde{\gamma} }{\gamma^2(1-\alpha)}.    
\end{equation}
Note that $u'_h(\gamma) = R_{i_j}$ and $u'_h(\tilde{\gamma}) = R_{i_{j_q}}$. Therefore, we can rewrite the derivative of $f$ at $\gamma$ as:
\begin{equation} \label{eqn:derivative_f_simple1}
f'(\gamma) = \frac{R_{i_j}\gamma - R_{i_{j_q}} \tilde{\gamma}}{R_{i_{j_q}} \gamma^2(1-\alpha)}.    
\end{equation}
Recall that $u_h(\tilde{\gamma}) = R_{i_{j_q}} \tilde{\gamma} - c_{i_{j_q}}$ is equal to $u_h(\gamma) - \kappa_S$ where $u_h(\gamma) = R_{i_j} \gamma - c_{i_j}$. As a result, we can simplify the numerator of \eqref{eqn:derivative_f_simple1} to $c_{i_j} - c_{i_{j_q}} + \kappa_S$. Thus, we have
\begin{equation} \label{eqn:derivative_f_simple2}
f'(\gamma) = \frac{c_{i_j} - c_{i_{j_q}} + \kappa_S}{R_{i_{j_q}} \gamma^2(1-\alpha)}.    
\end{equation}
Therefore, $f'(\gamma)$ is decreasing over $(\gamma_{j,q}, \gamma_{j,q+1})$, and hence, $f'(\cdot)$ is concave over this interval. Moreover, as we go from one interval to another, i.e., as $q$ increases, $R_{i_{j_q}}$ and $c_{i_{j_q}}$ both increase, meaning that $f'(\gamma)$ further decreases. This, along with the fact that minimum of concave functions is also concave, implies that $f(\cdot)$ is concave over the whole interval $[\gamma_j, \gamma_{j+1}]$. This completes our proof.  
\subsection{Proof of Theorem \ref{theorem:singleAgent_comparativestatics}}
\textbf{Proof of part (1):} Notice that changing $\kappa_I$ does not change the $\{\gamma_j\}_j$ and the $\beta(\cdot)$ function. As a consequence, changing $\kappa_I$ does not change the actions that $(\gamma^*, \beta^*)$ and $(\gamma'^*, \beta'^*)$ implement. That said, let us denote the reward of the actions implemented by $(\gamma^*, \beta^*)$ and $(\gamma'^*, \beta'^*)$ by $R$ and $R'$, respectively. Also, note that it suffices to show $\beta'^* \leq \beta^*$, and the other result $\gamma'^* \geq \gamma^*$ will be implied using the fact that $\beta(\cdot)$ is a decreasing function.

Suppose $\kappa_I$ is increased to to $\kappa'_I$. Since $(\gamma^*, \beta^*)$ is the optimal action with $\kappa_I$, we have
\begin{equation} \label{eqn:proof_comparative_1}
(1-\gamma^*) R - \beta^* \kappa_I \geq (1-\gamma'^*) R' - \beta'^* \kappa_I.    
\end{equation}
We prove the desired result by contradiction. Suppose $\beta'^* > \beta^*$. Hence, we have
\begin{equation}\label{eqn:proof_comparative_2}
-\beta^*(\kappa'_I - \kappa) > -\beta'^*(\kappa'_I - \kappa).    
\end{equation}
Adding the two sides of \eqref{eqn:proof_comparative_1} and \eqref{eqn:proof_comparative_2} implies
\begin{equation}
(1-\gamma^*) R - \beta^* \kappa'_I > (1-\gamma'^*) R' - \beta'^* \kappa'_I.    
\end{equation}
However, this is in contradiction with the assumption that $(\gamma'^*, \beta'^*)$ is an optimal contract when the cost of inspection is $\kappa'_I$. This completes the proof of this part.\\
\textbf{Proof of part (2):} Increasing $\kappa_S$ to $\kappa'_S$ does not change $\{\gamma_j\}$'s and the upper envelope function $u_h(\cdot)$, and therefore, does not change the actions that $(\gamma^*, \beta^*)$ and $(\gamma''^*, \beta''^*)$ implement. Let us denote the reward of the actions implemented by these two contracts by $R$ and $R''$, respectively. On the other hand, changing $\kappa_S$ changes the $\beta(\cdot)$ function to a new function $\hat{\beta}(\cdot)$. It is straightforward to see that $\hat{\beta}(\cdot)$ is pointwise larger than $\beta(\cdot)$.

Using the optimality of $(\gamma^*, \beta^*)$ with $\kappa_S$, we have
\begin{equation} \label{eqn:proof_comparative_3}
(1- \gamma^*)R - \beta(\gamma^*) \kappa_I \geq (1-\gamma''^*)R'' - \beta(\gamma''^*) \kappa_I.    
\end{equation}
The optimality of $(\gamma''^*, \beta''^*)$ with $\kappa'_S$ implies
\begin{equation}\label{eqn:proof_comparative_4}
(1-\gamma''^*)R'' - \hat{\beta}(\gamma''^*) \kappa_I \geq (1- \gamma^*)R - \hat{\beta}(\gamma^*) \kappa_I.    
\end{equation}
Summing the two sides of \eqref{eqn:proof_comparative_3} and \eqref{eqn:proof_comparative_4} and simplifying it gives us:
\begin{equation}\label{eqn:proof_comparative_5}
 \hat{\beta}(\gamma^*) - \beta(\gamma^*) \geq  \hat{\beta}(\gamma''^*) - \beta(\gamma''^*).
\end{equation}
We next make the following claim:
\begin{claim} \label{claim:pi_S_derivative}
$\hat{\beta}(\gamma) - \beta(\gamma)$ is a decreasing function of $\gamma$.     
\end{claim} 
Notice that this claim, along with \eqref{eqn:proof_comparative_5}, gives us the desired result. It remains to show that the claim holds. 

\textbf{Proof of Claim \ref{claim:pi_S_derivative}:} Recall from Lemma \ref{lemma:gamma_increasing} that both $\beta(\cdot)$ and $\hat{\beta}(\cdot)$ are decreasing functions of $\gamma$. Also, as stated above, for any $\gamma$, $\hat{\beta}(\gamma) \geq \beta(\gamma)$. Hence, we could divide $[0,1]$ to at most three intervals: (1) $[0,\gamma^{th}_1)$ where both $\beta(\gamma)$ and $\hat{\beta}(\gamma)$ are positive, (2) $[\gamma^{th}_1, \gamma^{th}_2)$ where $\beta(\gamma) = 0$ but $\hat{\beta}(\gamma)$ is positive, and (3) $[\gamma^{th}_2, 1]$ where both  $\beta(\gamma)$ and $\hat{\beta}(\gamma)$ are zero. 
We need to verify that Claim \ref{claim:pi_S_derivative} holds over all these three intervals. Note that since $\beta(\gamma)$ and $\hat{\beta}(\gamma)$ are continuous, once we have the claim established in all these three intervals, it also holds over the whole interval $[0,1]$. 

Obviously Claim \ref{claim:pi_S_derivative} holds for the third interval $[\gamma^{th}_2, 1]$ . It also holds over the second interval, i.e., $[\gamma^{th}_1, \gamma^{th}_2)$, since $\hat{\beta}(\gamma) - \beta(\gamma) = \hat{\beta}(\gamma)$ which is a decreasing function of $\gamma$. Hence, it remains to show that our claim holds over the first interval, i.e., when $\beta(\gamma)$ and $\hat{\beta}(\gamma)$ are both positive. 
In this case $\beta(\gamma)$ is given by
\begin{equation}
\beta(\gamma) = 1- \frac{u_h^{-1}(u_h(\gamma) - \kappa_S)}{\gamma(1-\alpha)}.    
\end{equation}
Therefore, we need to show that $g(\gamma, \kappa'_S) - g(\gamma, \kappa_S)$, with 
\begin{equation} 
g(\gamma, \kappa_S):= \frac{u_h^{-1}(u_h(\gamma) - \kappa_S)}{\gamma(1-\alpha)},
\end{equation}
is an increasing function of $\gamma$. Notice that $u_h^{-1}(\cdot)$ is a continuous function which is nondifferentiable at finitely many points. Hence, $g(\gamma, \kappa_S)$, as a function of $\kappa_S$, is continuous and nondifferentiable at all but finitely many points. Next, we claim 
\begin{equation} \label{eqn:g_pi_integral}
g(\gamma, \kappa'_S) - g(\gamma, \kappa_S) = \int_{\kappa_S}^{\kappa'_S} \frac{\partial}{\partial \kappa} g(\gamma, \kappa) d \kappa,     
\end{equation}
where the integral in \eqref{eqn:g_pi_integral} is the Lebesgue integral. To establish this result, we can partition the interval $[\kappa_S, \kappa'_S]$ into subintervals where $g(\gamma,\kappa)$ is differentiable with respect to to $\kappa$ over each of them. Subsequently, we can apply the fundamental theorem of calculus to each of these subintervals and then integrate them using the continuity of $g$ to derive \eqref{eqn:g_pi_integral}. Another way to establish \eqref{eqn:g_pi_integral} is to first recognize that $u_h^{-1}$ is an absolutely continuous function as it is a piecewise linear function. Then we use the generalized version of the fundamental theorem of calculus (see Theorem 7.18 in \cite{rudin1987real}).

Now, note that the derivative of $g$ with respect to $\kappa$ is given by
\begin{equation}
\frac{\partial}{\partial \kappa} g(\gamma, \kappa) = \frac{-1}{\gamma(1-\alpha)u_h(u_h^{-1}(u_h(\gamma)-\kappa_S))}.   
\end{equation}
Since $u_h(\cdot)$ is increasing function, it is straightforward to see $\frac{\partial}{\partial \kappa} g(\gamma, \kappa)$ is an increasing function of $\gamma$. This, along with \eqref{eqn:g_pi_integral}, shows that $g(\gamma, \kappa'_S) - g(\gamma, \kappa_S)$ is an increasing function of $\gamma$ which completes the proof of Claim \ref{claim:pi_S_derivative}.
{
\subsection{Proof of Lemma \ref{claim:error_discretization} in Theorem \ref{theorem:multi-agent-DP}}
It is sufficient to prove that for any $\ell$, $\mathcal{U}_p^\ell(\cdot)$ is Lipschitz with a parameter bounded by
\begin{equation*}
\frac{(R_n^\ell)^2}{\kappa^\ell_S}-\kappa_I^\ell.
\end{equation*}
By taking the optimal solution and rounding agent $\ell$'s inspection level down to the nearest element of the grid $\mathcal{G}^\ell$, the discretization error will be limited to $\delta \left [ \frac{(R_n^\ell)^2}{\kappa^\ell_S}-\kappa_I^\ell \right]$ given that the grid has a resolution of $\delta$. Note that by rounding down, we ensure the inspection budget is not exceeded.

Further, according to \eqref{eqn:principal_utility_bar_beta}, $\mathcal{U}_p^\ell(\beta)$ is either equal to $\mathcal{U}_p^\ell ( \gamma^\ell(\beta), \beta)$ or remains constant. Hence, to determine the upper bound on the Lipschitz parameter of $\mathcal{U}_p^\ell(\cdot)$, it is enough to find the upper bound for the derivative of $\mathcal{U}_p^\ell ( \gamma^\ell(\beta), \beta)$ as a function of $\beta$. While it might not be a continuously differentiable function, it is piecewise differentiable. Given its continuity, bounding its derivative across each segment suffices for our objective.
Recall from \eqref{eqn:U_p_ell_beta}
\begin{equation*}
\mathcal{U}_p^\ell ( \gamma^\ell(\beta), \beta) = (1-\gamma^\ell(\beta) )R_{i_j^\ell}^\ell - \beta \kappa^\ell_I,     
\end{equation*}
where $\gamma^\ell(\beta)$ is a decreasing function of $\beta$. Hence, when differentiable, we have
\begin{equation} \label{eqn:bound_d_U_beta}
\frac{d}{d\beta} ~ \mathcal{U}_p^\ell ( \gamma^\ell(\beta), \beta) \leq 
\left \vert \frac{d}{d\beta} ~ \gamma^\ell(\beta) \right \vert R_n^\ell - \kappa_I^\ell,
\end{equation}
where we used the fact that $R_n^\ell \geq R_{i_j^\ell}^\ell$. Next, using the inverse function theorem, we have
\begin{equation}
\left \vert \frac{d}{d\beta} ~ \gamma^\ell(\beta) \right \vert \leq \left \vert \frac{1}{f'^\ell(\gamma)} \right \vert_{\text{at } \gamma^\ell(\beta)},     
\end{equation}
where $f^\ell(\cdot)$ is defined for agent $\ell$ and similar to \eqref{eqn:def_f} in the proof of Lemma \ref{lemma:gamma_increasing} for the single-agent case. Note that, by \eqref{eqn:derivative_f_simple2}, we have
\begin{equation*}
f'^\ell(\gamma) \geq \frac{\kappa_S^\ell}{R_n^\ell \gamma^2 (1-\alpha)} 
\geq \frac{\kappa_S^\ell}{R_n^\ell},
\end{equation*}
which, along with \eqref{eqn:bound_d_U_beta} completes the proof of the lemma.
}
{ \subsection{Proof of Lemma \ref{lemma:scheduling}}
For any $b \in [B]$, define $\ell_b$ as the smallest integer $\ell$ such that $\sum_{j=1}^\ell \bar{\beta}^j \geq b$ (and let $\ell_0=1$ for convenience).

Inspectors' assignments are decided sequentially, progressing from inspector 1 to $B$. Specifically, inspector $b$ is assigned to agent $w_b \in \{\ell_{b-1}, \ell_{b-1}+1, \ldots, \ell_{b}\}$. Starting with the first inspector, they inspect agent $\ell < \ell_1$ with probability $\bar{\beta}^\ell$ and agent $\ell_1$ with probability $\zeta_1:= 1- \sum_{i=1}^{\ell_1-1} \bar{\beta}^i$. This assignment can be carried out using the uniform random variable generator as stated earlier.

Assuming we have determined assignments for inspectors $1, \ldots, b-1$, we next decide the agent for inspector $b$. A key challenge arises because agent $\ell_{b-1}$ is inspected with probability $\zeta_{b-1} := b-1 - \sum_{i=1}^{\ell_{b-1}-1} \bar{\beta}^i$, which could be less than its targeted inspection probability $\bar{\beta}^{\ell_{b-1}}$. Therefore, inspector $b$ should inspect agent $\ell_{b-1}$ with the remaining probability $\bar{\beta}^{\ell_{b-1}} - \zeta_{b-1}$, but this should only occur when inspector $b-1$ has opted to inspect other agents. To do so, we use the following procedure:
\begin{enumerate}[i)]
\item If $w_{b-1} = \ell_{b-1}$, then inspector $b$ inspects one of the agents in the set $\{\ell_{b-1}+1, \ldots, \ell_{b}\}$. In other words, in this scenario, inspector $b$'s single unit of inspection is distributed among these agents. In particular, for agent $\ell < \ell_{b}$, their probability of inspection is proportional to $\bar{\beta}^\ell$, and for agent $\ell_b$, it is proportional to $\zeta_b$. More formally, agent $\ell \in \{\ell_{b-1}+1, \ldots, \ell_{b}-1\}$ is inspected with probability
\begin{equation*}
\frac{\bar{\beta}^\ell}{1-\bar{\beta}^{\ell_{b-1}} + \zeta_{b-1}},    
\end{equation*}
and agent $\ell_b$ is inspected with probability 
\begin{equation*}
\frac{\zeta_b}{1-\bar{\beta}^{\ell_{b-1}} + \zeta_{b-1}}.
\end{equation*}
\item If $w_{b-1} \neq \ell_{b-1}$, then agent $\ell_{b-1}$ is inspected by probability 
\begin{equation*}
\frac{\bar{\beta}^{\ell_{b-1}} - \zeta_{b-1}}{1 - \zeta_{b-1}}.  
\end{equation*}
The remaining probability of inspection in this case is divided between agents $\{\ell_{b-1}+1, \ldots, \ell_{b}\}$ in a proportional way similar in part (i).
\end{enumerate}
The above procedure ensures that agent $\ell_{b-1}$ is not inspected by both inspectors $b-1$ and $b$. Also, it is straightforward to verify that each agent $\ell$'s probability of inspection matches the given desired level $\bar{\beta}^{\ell_{b-1}}$. This completes the proof.}
{
\subsection{Proof of Proposition \ref{proposition:illustrative}}
For simplicity of notation, assume that $m$ is even. All results proceed similarly if $m$ is odd. We denote the high and low inspection costs by $\kappa_I^H$ and $\kappa_I^L$, respectively. Without loss of generality, let us assume the inspection cost corresponding to the first $m/2$ agents is $\kappa_I^L$ and the inspection cost corresponding to the second half is $\kappa_I^H$.

Since agents differ only in their inspection costs, we deduce that
\begin{equation*}
\mathcal{U}_p^1 ( \gamma^1(\beta), \beta) = \cdots =  \mathcal{U}_p^{m/2} ( \gamma^{m/2}(\beta), \beta) >  \mathcal{U}_p^{m/2+1} ( \gamma^{m/2+1}(\beta), \beta) = \cdots = \mathcal{U}_p^m ( \gamma^m(\beta), \beta),
\end{equation*}
for any $\beta > 0$. This leads directly to
\begin{equation} \label{eqn:order_U_p}
\mathcal{U}_p^1(\beta) = \cdots = \mathcal{U}_p^{m/2}(\beta)
> \mathcal{U}_p^{m/2+1}(\beta) = \cdots = \mathcal{U}_p^m(\beta)    
\end{equation}
for any $\beta > 0$.

Let $(\bar{\beta}^1, \cdots, \bar{\beta}^m)$ be a solution to the optimization problem \eqref{eqn:principal_problem_multiAgent}. Using \eqref{eqn:order_U_p}, we infer $\bar{\beta}^i \geq \bar{\beta}^j$ for any $i \in [m/2]$ and $j \in \{m/2+1, \cdots, m\}$. The inequality becomes strict when the values are positive. To see this, note that, otherwise, the principal could swap $\bar{\beta}^i$ with $\bar{\beta}^j$ to increase their utility.

Without loss of generality, let us assume that $\bar{\beta}^{m/2}$ is the minimum inspection level among the first half of the agents, i.e., $\bar{\beta}^{m/2} = \min {\bar{\beta}^1, \cdots, \bar{\beta}^{m/2}}$, and $\bar{\beta}^{m/2+1}$ is the maximum inspection level among the second half, i.e., $\bar{\beta}^{m/2+1} = \max {\bar{\beta}^{m/2+1}, \cdots, \bar{\beta}^{m}}$. As previously discussed, $\bar{\beta}^{m/2} \geq \bar{\beta}^{m/2+1}$, and the inequality is strict if the value on the left-hand side is positive.

Now, if $\bar{\beta}^{m/2+1}=0$, then we are done. Otherwise, we have 
\begin{equation*}
\frac{2B}{m} \geq \bar{\beta}^{m/2} > \bar{\beta}^{m/2+1} > 0.    
\end{equation*}
Given that $\mathcal{U}_p^\ell(\cdot)$ is piecewise differentiable, for sufficiently large $m$, both $\bar{\beta}^{m/2}$ and $ \bar{\beta}^{m/2+1}$ will fall in the first differentiable segments of $\mathcal{U}_p^{m/2}(\cdot)$ and $\mathcal{U}_p^{m/2+1}(\cdot)$ respectively. As a result, and given that agents are similar aside from their inspection costs, for $l \in \{m/2, m/2+1\}$ we have
\begin{align} \label{eqn:derivative:first_segment}
\frac{d}{d\beta} ~ \mathcal{U}_p^{l} (\bar{\beta}^{l}) =
- \frac{d}{d\beta} ~ \gamma(\beta) R - \kappa_I^\ell,
\end{align}
for some $R$ and $\gamma(\cdot)$. 

On the other hand, note that we have
\begin{equation} \label{eqn:equal_derivatives}
\frac{d}{d\beta} ~ \mathcal{U}_p^{m/2} (\bar{\beta}^{m/2}) =      
\frac{d}{d\beta} ~ \mathcal{U}_p^{m/2+1} (\bar{\beta}^{m/2+1}), 
\end{equation}
because if this weren't the case, a slight increase in the inspection level of the one with the higher derivative, along with the decrease by the same amount in the one with the lower derivative, would boost the total principal utility. (The Karush–Kuhn–Tucker (KKT) conditions can also be used to arrive at this result, as, in fact, the proof of the KKT theorem uses a rationale analogous to the one proposed here~\citep[cf.][]{bertsekas1997nonlinear}).

Putting \eqref{eqn:derivative:first_segment} and \eqref{eqn:equal_derivatives} together, we have
\begin{equation}
\kappa_I^{H} - \kappa_I^{L} =
R \left ( \frac{d}{d\beta} ~ \gamma(\beta^{m/2}) R - \frac{d}{d\beta} ~ \gamma(\beta^{m/2+1}) R \right ).
\end{equation}
Notice that, as $m$ grows, the right-hand side goes to zero as $\frac{d}{d\beta} ~ \gamma(\cdot) R$ is a continuous function and both $\beta^{m/2}, \beta^{m/2+1} \in [0, 2B/m]$. However, the left-hand side remains constant which leads to a contradiction. This completes the proof.
}
\end{document}